\documentclass[10pt,a4paper]{article}
\usepackage[utf8]{inputenc}
\usepackage{amsmath, amsthm, amssymb}
\usepackage{enumerate}
\usepackage[svgnames]{xcolor}
\usepackage[toc,page]{appendix}
\usepackage[a4paper]{geometry}
\usepackage{tikz-cd}
\newtheorem{thm}{Theorem}[section]
\newtheorem{lem}[thm]{Lemma}
\newtheorem{tvrz}[thm]{Proposition}
\newtheorem{lemma}[thm]{Lemma}

\newtheorem{theorem}[thm]{Theorem}
\theoremstyle{definition}
\newtheorem{definice}[thm]{Definition}
\theoremstyle{remark}
\newtheorem{rem}[thm]{Remark}
\theoremstyle{definition}
\newtheorem{example}[thm]{Example}

\def\R{\mathbb{R}}
\def\C{\mathcal{C}}

\def\T{\mathcal{T}}
\def\<{\langle}
\def\>{\rangle}
\def\dal{\langle \! \langle}
\def\dar{\rangle \! \rangle}
\def\~{\widetilde}
\def\^{\wedge}
\def\ddt{\left. \frac{d}{dt}\right|_{t=0} \hspace{-0.5cm}}

\def\g{\mathfrak{g}}
\def\h{\mathfrak{h}}

\def\n{\mathfrak{n}}

\def\aso{\mathfrak{so}}
\def\ak{\mathfrak{k}}

\def\io{\mathit{i}}
\def\F{\mathcal{F}}
\def\G{\mathcal{G}}
\def\K{\mathcal{K}}
\def\B{\mathcal{B}}
\def\H{\mathcal{H}}
\def\A{\mathcal{A}}
\def\D{\mathcal{D}}
\def\M{\mathcal{M}}

\def\tr{\triangleright}

\def\tC{\tilde{C}}
\def\sC{\mathcal{CS}}

\def\gTM{\mathbb{T}M}

\def\gTP{\mathbb{T}P}

\def\gm{\mathbf{G}}

\def\RS{\mathcal{R}}
\def\fR{\mathbf{R}}

\def\fPsi{\mathbf{\Psi}}

\def\cD{\nabla}

\def\hcD{\widehat{\nabla}}
\def\hcDL{\widehat{\nabla}^{LC}}

\def\tcD{\widetilde{\nabla}}

\def\cDL{\nabla^{LC}}
\def\cDN{\nabla^{0}}

\def\cif{C^{\infty}(M)}

\newcommand{\bm}[4]{\begin{pmatrix} #1 & #2 \\ #3 & #4 \end{pmatrix}}

\newcommand{\vf}[1]{ \mathfrak{X}^{#1}(M)}
\newcommand{\df}[1]{ \Omega^{#1}(M)}
\newcommand{\vfP}[1]{ \mathfrak{X}^{#1}(P)}
\newcommand{\dfP}[1]{ \Omega^{#1}(P)}

\newcommand{\Li}[1]{ \mathcal{L}_{#1}}

\DeclareMathOperator{\BlockDiag}{BlockDiag}

\DeclareMathOperator{\vol}{vol}
\DeclareMathOperator{\End}{End}
\DeclareMathOperator{\gSO}{SO}
\DeclareMathOperator{\gE}{E}
\DeclareMathOperator{\rank}{rank}
\DeclareMathOperator{\Hom}{Hom}
\DeclareMathOperator{\LC}{LC}

\DeclareMathOperator{\Aut}{Aut}

\DeclareMathOperator{\Ric}{Ric}

\DeclareMathOperator{\Lie}{Lie}

\DeclareMathOperator{\Tr}{Tr}

\DeclareMathOperator{\Div}{div}
\DeclareMathOperator{\hRic}{\widehat{R}ic}

\begin{document}
\begin{flushright}
\today
\end{flushright}
\vspace{0.7cm}
\begin{center}

\baselineskip=13pt {\Large \bf{Kaluza-Klein Reduction of Low-Energy Effective Actions: Geometrical Approach}\\}
 \vskip0.5cm
 {\it Dedicated to our son on the occasion of his first birthday}  
 \vskip0.7cm
 {\large{ Jan Vysoký}}\\
 \vskip0.6cm
\textit{Institute of Mathematics of the Czech Academy of Sciences \\ Žitná 25, Prague 11567, Czech Republic, vysoky@math.cas.cz}\\
\vskip0.3cm
\end{center}

\begin{abstract}
Equations of motion of low-energy string effective actions can be conveniently described in terms of generalized geometry and Levi-Civita connections on Courant algebroids. This approach is used to propose and prove a suitable version of the Kaluza-Klein-like reduction. Necessary geometrical tools are recalled. 
\end{abstract}

{\textit{Keywords}: String low-energy effective actions, Kaluza-Klein reduction, Generalized geometry, Courant algebroids, Levi-Civita connections, Heterotic supergravity}. 
\section{Introduction}
This paper serves as a thorough extension of ideas sketched in the previous paper \cite{Jurco:2015bfs} written together with Branislav Jurčo. In particular, it turned out that instead of a direct relation between low-energy string effective actions, we should aim for relation of the respective equations of motion. Moreover, it proved necessary to use a different and technically better definition of the generalized Riemann tensor inspired by the one used in double field theory \cite{Hohm:2012mf}. Using those polished tools of generalized geometry, we are able to give a full reduction procedure of a low-energy string effective action. 

The original idea \cite{kaluza1921unitatsproblem} of Kaluza to obtain a unification of four-dimensional gravity with  electromagnetism via the reduction of five-dimensional gravity dates back to year 1919 (published in 1921). Five years later it was examined by Klein from the quantum point of view in \cite{klein1926quantentheorie}, thus becoming known as Kaluza-Klein (KK) theory (or reduction). Since then, KK theory played an important role throughout the development of any candidate for the grand unified theory, in particular in string theory. It is beyond the scope of this paper to give a full list of references. Instead, we recommend the historical review in \cite{Duff:1994tn} or a very  comprehensive one in \cite{Bailin:1987jd}.

In this paper, we attempt to view this from the perspective within the framework of generalized geometry as introduced by Hitchin \cite{Hitchin:2004ut} and further developed in \cite{Gualtieri:2003dx, 2005math......8618H,2006CMaPh.265..131H}. In particular, it turned out that one should examine the natural generalizations of Levi-Civita connections with respect to a generalized Riemannian metric. This idea is rather new, coming notably in the work of Garcia-Fernandez in \cite{2013arXiv1304.4294G, Garcia-Fernandez:2016azr, Garcia-Fernandez:2016ofz} and also in our papers with Braňo Jurčo in \cite{Jurco:2015xra, Jurco:2015bfs} and in detailed lecture notes \cite{Jurco:2016emw}. 

Let us now comment on the following. At first glance, the use of Courant algebroid connections to obtain the correspondence of equations of motion may seem unnecessary. Without any doubt, the equivalence of both theories can be obtained by hand. However, we believe that the approach taken here is novel and useful as it provides the guiding principle throughout the entire procedure. All assumptions are made by considering the natural restrictions of the involved structures given by requirements of symmetry. For example, straightforward reducibility impositions on the generalized metric on the larger spacetime provide the correct requirements on the fields of the low-energy string effective action. Moreover, by construction, one a priori knows that the two sets of equations of motion are \emph{somehow} related and all calculations are then basically just a verification. Last but not least, we believe that the framework here is quite universal and it would allow for future straightforward generalizations by considering more complicated examples of involved structures e.g. to include the fermionic sector of supergravity.

The basic idea is the following. As we aim for non-Abelian gauge theories, we consider a principal $G$-bundle $\pi: P \rightarrow M$ with a compact semisimple structure Lie group $G$. Equations of motion of the string low-energy effective action (or equivalently the bosonic part of type IIB supergravity) can be geometrically described in terms of the (possibly $H$-twisted) Dorfman bracket on the generalized tangent bundle $\gTP \equiv TP \oplus T^{\ast}P$ and the generalized metric $\gm$ corresponding to the metric $g$ and the Kalb-Ramond $2$-form $B$. Interestingly, the scalar function $\phi$ corresponding to the dilaton field originates from the freedom in the construction of a metric-compatible and torsion-free Courant algebroid connection on $\gTP$. 

Under certain conditions on the twisting $3$-form $H$, $\gTP$ becomes a so called equivariant Courant algebroid, allowing for its reduction to a Courant algebroid $E'$ above the base manifold $M$.  See \cite{Bursztyn2007726} and \cite{Baraglia:2013wua}. The resulting structure is called the heterotic Courant algebroid, and its relevance is due to the condition on the triviality of the first Pontryagin class. This is exactly the Green-Schwarz anomaly cancellation condition when the principal bundle $P$ is a fibre product of a Yang-Mills bundle and the spinor bundle on $M$. See \cite{Baraglia:2013wua, 2013arXiv1304.4294G}. Hence, the structure of a heterotic Courant algebroids can be used to naturally incorporate the corresponding $\alpha'$ correction. In particular, one can encode a theory resembling the bosonic part of heterotic supergravity in terms of the generalized Riemann tensor on $E'$. 

It is natural to examine generalized metrics $\gm$ on $\gTP$ which reduce naturally onto the generalized metrics $\gm'$ on the reduced heterotic Courant algebroid $E'$ and similarly  Levi-Civita connections $\cD$ with respect to $\gm$ that in some sense reduce to Levi-Civita connections $\cD'$ on $E$ with respect to $\gm'$. One could then expect the low-energy effective actions described by both connections to be related to each other. In this paper we find conditions under which they are completely (classically) equivalent. See Theorem \ref{thm_main}. Note that on the level of differential equations, this is was already observed in \cite{Baraglia:2013wua}. 

The heterotic effective actions, Green-Schwarz mechanism and the related $\alpha'$ correction have been extensively examined recently in the double field theory \cite{Bedoya:2014pma, Hohm:2014xsa, Hohm:2015mka, Hohm:2014eba, Marques:2015vua}. For a general review of double field theory, including discussion of effective action see \cite{Hohm:2013bwa, Aldazabal:2013sca}. Recall also the generalized geometry approach to $\alpha'$ corrections published in \cite{Coimbra:2014qaa}, and generalized connections with applications in DFT in \cite{Jeon:2010rw, Jeon:2011cn}. Using the Riemann-like tensors encoding the low-energy effective actions date back to Siegel \cite{Siegel:1993bj, Siegel:1993th}. In DFT, generalized Riemann tensors were examined extensively in \cite{Hohm:2010xe} and in \cite{Hohm:2011ex} for heterotic case. For more geometrical approach, see \cite{Hohm:2012mf}. Note that an important role of Courant algebroids and generalized geometry in supergravity was conjectured in the talk of Peter Bouwknegt \cite{talkbouwknegt}. There are many recent developments of similar ideas, see in particular the work of Coimbra, Strickland-Constable and Waldram in \cite{Coimbra:2011nw,Coimbra:2012af}. 

This paper is organized as follows:

In Section \ref{sec_KKR}, we establish the terminology, notation and we formulate the main theorem of this paper. This part does not require any knowledge of generalized geometry tools used in the proof of the statement. 

In Section \ref{sec_EOM}, we employ the variational principle to derive the equations of motion from the respective action functionals. Of course, as this is nothing new, an uninterested reader can skip to the end of this part and take note only of Theorem \ref{thm_EOM} to get acknowledged with the notation. 

Section \ref{sec_Courant} serves as a quick introduction to the theory of Courant algebroids and Levi-Civita connections. Our intention is to skip all unnecessary details and proofs, as we have already addressed this in detail in \cite{Jurco:2016emw}.  

The main goal of Section \ref{sec_EOMgeom} is to provide a detailed and explicit calculation leading to the description of equations of motion in terms of generalized Riemann tensor on the heterotic Courant algebroid $E'$. This should repay our debt from the previous paper \cite{Jurco:2015bfs} where we have hidden (mainly because they were much more complicated and crude at the time) all calculations from the reader. One can also skip to the resulting Theorem \ref{thm_eomheterotic}. 

In Section \ref{sec_reduction}, we quickly review the reduction of equivariant Courant algebroids, generalized metrics and Levi-Civita connections. Most importantly, in Subsection \ref{subsec_comparison} we derive a crucial relation of the quantities required to describe the equations of motion.  

Finally, we are able to provide a proof of the main theorem in Section \ref{sec_proof}. At this point, it is just a combination of already proved theorems of previous two sections. 

We conclude this paper by Section \ref{sec_analysis} where we relate our notation using the globally defined objects to the more conventional notation using the local connection and curvature forms. In particular, we find a direct relation to the bosonic content of heterotic supergravity for the particular choice of the principal bundle $P$. 
\section{Kaluza-Klein reduction} \label{sec_KKR}
Let us formulate the main theorem of this paper. Let $\pi: P \rightarrow M$ be a principal $G$-bundle with a semisimple and compact Lie group $G$, and let $\g = \Lie(G)$. Let $c = \<\cdot,\cdot\>_{\g}$ denote the corresponding negative-definite Killing form.

Let $B \in \dfP{2}$ be a $2$-form on $P$ usually called a \textbf{Kalb-Ramond field}. Let $H \in \dfP{3}$ be any $3$-form, and let $H' = H + dB$. 
Assume that $g$ is a metric on $P$. At this moment it can have any signature. Let $\phi \in C^{\infty}(P)$ be a scalar function called a \textbf{dilaton field}, and let $\Lambda \in \R$ be a real number called a \textbf{cosmological constant}.

Let $d\vol_{g}$ denote the volume form corresponding to $g$, and let $\<\cdot,\cdot\>_{g}$ be a fiber-wise scalar product of $p$-forms induced by the Hodge duality operator $\ast_{g}$, that is $\alpha \^ \ast_{g}(\beta) = \<\alpha,\beta\>_{g} \cdot d\vol_{g}$. Let $\RS(g)$ be the scalar curvature of the Levi-Civita connection corresponding to $g$. We can now define the following action functional: 
\begin{equation} \label{eq_S}
S[g,B,\phi] = \int_{P} e^{-2 \phi} \{ \RS(g) - \frac{1}{2} \<H',H'\>_{g} + 4 \< d\phi, d\phi\>_{g} - 2 \Lambda \} \cdot d\vol_{g}.
\end{equation}
We neglect any constants or string parameters $\alpha'$. We will derive the equations of motion for this action in the following section. $S$ can be thus viewed as a field theory for backgrounds $(g,B,\phi)$ living on the principal bundle total space $P$. 

On the other hand, let $B_{0} \in \df{2}$ be a $2$-form on the base manifold, and let $H_{0} \in \df{3}$ be any $3$-form. Let $g_{0}$ be a metric on $M$, and let $\vartheta \in \Omega^{1}(M,\g_{P})$ be a $1$-form on $M$ valued in the adjoint bundle $\g_{P}$. Let $\phi_{0} \in \cif$ be a scalar dilaton field on $M$, and let $\Lambda_{0} \in \R$ be a cosmological constant. Let $\dal \cdot , \cdot \dar$ be a fiber-wise scalar product induced by combination of $\<\cdot,\cdot\>_{g_{0}}$ and $\<\cdot,\cdot\>_{\g}$ on $\Omega^{p}(M,\g_{P})$, and let $D: \Omega^{\bullet}(M,\g_{P}) \rightarrow \Omega^{\bullet+1}(M,\g_{P})$ be an exterior covariant derivative induced by the fixed principal bundle connection $A \in \Omega^{1}(P,\g)$. Define $F' \in \Omega^{2}(M,\g_{P})$ and $H'_{0} \in \df{3}$:
\begin{equation} \label{eq_F'andH'0}
F' = F + D\vartheta + \frac{1}{2} [\vartheta \^ \vartheta]_{\g}, \; \; H'_{0} = H_{0} + dB_{0} - \frac{1}{2} \tC_{3}(\vartheta) - \< F \^ \vartheta \>_{\g}, 
\end{equation}
where $\tC_{3}(\vartheta) = \< D\vartheta \^ \vartheta \>_{\g} + \frac{1}{3} \<[\vartheta \^ \vartheta]_{\g} \^ \vartheta \>_{\g}$ is a "Chern-Simons like" $3$-form on $M$, and $F \in \Omega^{2}(M,\g_{P})$ is the curvature $2$-form of $A$. We propose a following action functional:
\begin{equation} \label{eq_S0}
S_{0}[g_{0},B_{0},\phi_{0},\vartheta] = \int_{M} e^{-2\phi_{0}} \{ \RS(g_{0}) + \frac{1}{2} \dal F', F' \dar - \frac{1}{2} \<H'_{0},H'_{0}\>_{g_{0}} + 4 \<d\phi_{0},d\phi_{0}\>_{g_{0}} - 2 \Lambda_{0} \} \cdot d\vol_{g_{0}}. 
\end{equation}
It is a partial goal of this paper to explain how to obtain exactly this combination of the dynamical fields $(g_{0},B_{0},\phi_{0},\vartheta)$. 

Recall that any principal bundle connection $A \in \Omega^{1}(P,\g)$ provides a decomposition $\mathfrak{X}_{G}(P) \cong \mathfrak{X}(M) \oplus \Gamma(\g_{P})$ of $G$-invariant vector fields, given by $X = V^{h} + j(\Phi)$, where $V^{h}$ denotes the horizontal lift of a vector field $V \in \vf{}$, and $j(\Phi)$ is for each $\Phi \in \Gamma(\g_{P}) = C^{\infty}_{Ad}(P,\g)$ defined as $\{ j(\Phi)\}(p) = \#\{ \Phi(p) \}$ for all $p \in P$, where $x \in \g \mapsto \#{x} \in \vfP{}$ denotes the infinitesimal Lie algebra action of $\g$ on $P$.

Having explained the notation of this paper, we may proceed to the main theorem of this work. Its proof will be provided in the following sections. 
\begin{theorem}[\textbf{Kaluza-Klein reduction of a low-energy effective action}] \label{thm_main}
Assume that $g$, $B$ and $\phi$ are $G$-invariant tensor fields on $P$. Let $A \in \Omega^{1}(P,\g)$ be a fixed principal bundle connection on $P$. In particular $\phi = \phi_{0} \circ \pi$ for $\phi_{0} \in \cif$, and $g,B$ can be written in block form with respect to the decomposition $TM \oplus \g_{P}$. Furthermore, assume that their formal block matrices have the form:
\begin{equation} \label{eq_gBrelevant}
g = \bm{1}{\vartheta^{T}}{0}{1} \bm{g_{0}}{0}{0}{-\frac{1}{2}c} \bm{1}{0}{\vartheta}{1} , \; \; B = \bm{B_{0}}{\frac{1}{2} \vartheta^{T}c}{- \frac{1}{2} c \vartheta}{0},
\end{equation}
where $g_{0}$ is a metric on $M$, $B_{0} \in \df{2}$, and $\vartheta \in \Omega^{1}(M,\g_{P})$. Let $\Lambda = \Lambda_{0} + \frac{1}{6} \dim{\g}$, and fix the $3$-form $H$ to be $H = \pi^{\ast}(H_{0}) + \frac{1}{2} CS_{3}(A)$, where $CS_{3}(A)$ is a Chern-Simons $3$-form corresponding to the connection $A$. 

Then $(g,B,\phi)$ satisfy the equations of motion given by the action functional (\ref{eq_S}), if and only if $(g_{0},B_{0},\phi_{0},\vartheta)$ satisfy the equations of motion given by the action function (\ref{eq_S0}). In other words, the field theories given by $S$ and $S_{0}$ are \textbf{classically equivalent}. 
\end{theorem}
\begin{rem}
We should now comment on the term "classical equivalence". The full theory on the spacetime $P$ is of course not completely equivalent to the one on the orbit space $M$. Therefore, by the classical equivalence of $S$ and $S_{0}$ we mean that equations of motion given by the extremalization of the functional $S$ \emph{together} with the equations restricting the background fields $(g,B,\phi)$ to be $G$-invariant and have the form (\ref{eq_gBrelevant}) are satisfied, if and only if the corresponding background fields $(g_{0},B_{0},\phi_{0},\vartheta)$ on $M$ extremalize the action functional $S_{0}$. 
\end{rem}
\section{Equations of motion} \label{sec_EOM}
In this section, we will derive the equations of motion coming from actions (\ref{eq_S}) and (\ref{eq_S0}). We include this section mainly to establish the notation. The variations of the respective actions are well-known results and can be skipped by any reader acquainted
by these classical results. We write this in the form of a sequence of lemmas and propositions, summarizing the statements in this section with Theorem \ref{thm_EOM}. 
\begin{lem}
Let $g$ be a metric on a manifold $M$, and fix a scalar function $\phi \in \cif$. Let 
\begin{equation}
S_{E}[g] = \int_{M} e^{-2\phi} \RS(g) \cdot d \vol_{g}. 
\end{equation}
Let $g' = g + \epsilon h$, where $h$ is any symmetric form on $M$ vanishing on $\partial M$ and $\epsilon \in \R$ is a real number small enough for $g'$ to be a metric on $M$. Then there holds an equation
\begin{equation}
S_{E}[g'] = S_{E}[g] - \epsilon \int_{M} h^{mn} G^{\phi}_{mn} \cdot d\vol_{g} + o(\epsilon^{2}) ,
\end{equation}
where $G^{\phi}$ is a \textbf{modified Einstein tensor} given by
\begin{equation}
\begin{split}
G^{\phi}(X,Y) = & \ e^{-2\phi}\{ \Ric^{LC}(X,Y) + (\cDL_{X}(d\phi))(Y) + (\cDL_{Y}(d\phi))(X) - 4  (X.\Phi) (Y.\Phi) \} \\
& + e^{-2\phi} \{ 4 \< d\phi,d\phi\>_{g} - 2 \Delta_{g}(\phi) - \frac{1}{2} \RS(g) \} \cdot g(X,Y),
\end{split}
\end{equation}
for all $X,Y \in \vf{}$, where $\cDL$ is the Levi-Civita connection corresponding to $g$, $\Ric^{LC}$ is its Ricci tensor and $\Delta_{g} = - \{ d \delta_{g} + \delta_{g} d \}$ is the Laplace-Bertrami operator. 
\end{lem}
For $\phi = 0$, one obtains an ordinary Einstein tensor. The most complicated part of the proof of this lemma in fact comes from the presence of the non-constant dilaton field $\phi$. To deal with the remaining terms in the action $S$ or $S_{0}$, one uses the following observation. This is again a standard result and we thus omit its proof:
\begin{lemma}
Let $\alpha \in \df{p}$, and let $g' = g + \epsilon h$ be two metrics as in the previous lemma. Then there holds an equation
\begin{equation}
\int_{M} \<\alpha,\alpha\>_{g'} \cdot d\vol_{g'} = \int_{M} \<\alpha,\alpha\>_{g} \cdot d\vol_{g} - \epsilon \int_{M} \frac{1}{2} h^{mn} t_{mn} \cdot d\vol_{g} + o(\epsilon^{2}), 
\end{equation}
where $t$ is a tensor defined as $t(X,Y) = 2 \< \io_{X}\alpha, \io_{Y}\alpha\>_{g} - \<\alpha,\alpha\>_{g} \cdot g(X,Y)$, for all $X,Y \in \vf{}$. The statement holds also for $\alpha \in \Omega^{p}(M,\g_{P})$ with $\<\cdot,\cdot\>_{g}$ replaced with $\dal \cdot,\cdot \dar$. 
\end{lemma}
Using these two lemmas, it is easy to prove the following two propositions:
\begin{tvrz} \label{tvrz_EOMg}
Let $g' = g + \epsilon h$, where $h$ is any symmetric form on $P$ vanishing on $\partial P$, and $\epsilon \in \R$ is a real number small enough for $g'$ to be a metric on $P$. Then for fixed $B \in \dfP{2}$ and $\phi \in C^{\infty}(P)$, one obtains 
\begin{equation}
S[g',B,\phi] = S[g,B,\phi] - \epsilon \int_{P} h^{mn} e^{-2\phi} \{ \beta(g)_{mn} - \frac{1}{2} \beta(\phi) g_{mn} \} \cdot d\vol_{g} + o(\epsilon^{2}),
\end{equation}
where $\beta(g)$ is a symmetric tensor field on $P$ defined by 
\begin{equation}
\beta(g)(X,Y) = \Ric^{LC}(X,Y) - \frac{1}{2}\<\io_{X}H',\io_{Y}H'\>_{g} + (\cDL_{X}(d\phi))(Y) + (\cDL_{Y}(d\phi))(X),
\end{equation}
for all $X,Y \in \vfP{}$, and $\beta(\phi)$ is a scalar field on $P$ defined by 
\begin{equation} \label{def_Bphi}
\beta(\phi) = \RS(g) - \frac{1}{2} \<H',H'\>_{g} + 4 \Delta_{g}(\phi) - 4 \<d\phi,d \phi\>_{g} - 2 \Lambda. 
\end{equation}
 Here $h^{mn}$ are components of $h$ with indices raised by $g$. 
\end{tvrz}

\begin{tvrz} \label{tvrz_EOMg0}
Let $g'_{0} = g_{0} + \epsilon h_{0}$, where $h_{0}$ is any symmetric form on $M$ vanishing on $\partial M$, and $\epsilon \in \R$ is a real number small enough for $g'_{0}$ to be a metric on $M$. Then for fixed $B_{0} \in \df{2}$, $\phi_{0} \in C^{\infty}(M)$ and $\vartheta \in \Omega^{1}(M,\g_{P})$, one obtains 
\begin{equation}
\begin{split}
S_{0}[g'_{0},B_{0},\phi_{0},\vartheta] = & \ S_{0}[g_{0},B_{0},\phi_{0},\vartheta] \\
 & - \epsilon \int_{M} h_{0}^{mn} e^{-2\phi_{0}} \{ \beta'(g_{0})_{mn} - \frac{1}{2} \beta'(\phi_{0}) (g_{0})_{mn} \} \cdot d\vol_{g_{0}} + o(\epsilon^{2}),
\end{split}
\end{equation}
where $\beta'(g_{0})$ is a symmetric tensor field on $M$ defined by 
\begin{equation}
\begin{split}
\beta'(g_{0})(X,Y) = & \ \Ric^{LC}_{0}(X,Y) + \frac{1}{2} \dal \io_{X}F', \io_{Y}F' \dar - \frac{1}{2}\<\io_{X}H'_{0},\io_{Y}H'_{0}\>_{g_{0}} \\
& + (\cD^{0}_{X}(d\phi_{0}))(Y) + (\cD^{0}_{Y}(d\phi_{0}))(X),
\end{split}
\end{equation}
for all $X,Y \in \vf{}$, and $\beta'(\phi_{0})$ is a scalar field on $M$ defined by 
\begin{equation}  \label{def_B'phi0}
\beta'(\phi_{0}) = \RS(g_{0}) + \frac{1}{2} \dal F',F' \dar - \frac{1}{2} \<H'_{0},H'_{0}\>_{g_{0}} + 4 \Delta_{g_{0}}(\phi_{0}) - 4 \<d\phi_{0},d \phi_{0}\>_{g_{0}} - 2\Lambda_{0}. 
\end{equation}
Here $h_{0}^{mn}$ are components of $h_{0}$ with indices raised by $g_{0}$. $\cD^{0}$ denotes the Levi-Civita connection corresponding to $g_{0}$, $\Ric^{LC}_{0}$ is the corresponding Ricci tensor, and its trace using $g_{0}$ is denoted as $\RS(g_{0})$. 
\end{tvrz}
We will now show that equations of motion for scalar dilation fields corresponds to the vanishing of the functions $\beta(\phi)$ and $\beta'(\phi')$ appearing in the previous two propositions. To do so, we will recall the following result of the Riemannian geometry:
\begin{lemma} 
Let $d$ be a de Rham differential, and let $\delta_{g}$ be a corresponding codiferential induced by $g$. Then under the integration sign, those operations are mutually adjoint:
\begin{equation} \label{eq_dcodadjoint}
\int_{M} \< d\alpha, \beta \>_{g} \cdot d\vol_{g} = \int_{M} \< \alpha, \delta_{g} \beta \>_{g} \cdot d\vol_{g},
\end{equation}
for all $\alpha \in \df{p}$ and $\beta \in \df{p+1}$, whenever either of the two forms vanishes on the boundary. Similarly, there holds the formula
\begin{equation} \label{eq_Dcodadjoint}
\int_{M} \dal D\alpha, \beta \dar \cdot d\vol_{g_{0}} = \int_{M} \dal \alpha, D^{\dagger}_{g_{0}}\beta \dar \cdot d\vol_{g_{0}},
\end{equation}
for all $\alpha \in \Omega^{p}(M,\g_{P})$, $\beta \in \Omega^{p+1}(M,\g_{P})$. Here $D$ is the exterior covariant derivative, and $D^{\dagger}_{g_{0}}$ the corresponding covariant codifferential, defined by the formula $D^{\dagger}_{g_{0}}(\alpha) = (-1)^{|\alpha|} \{ \ast_{g_{0}}^{-1} D \ast_{g_{0}}\}(\alpha)$. Again, either of the forms has to vanish on $\partial M$. 
\end{lemma}

We can use this lemma to prove the proposition examining the dilatonic equations of motion.
\begin{tvrz} \label{tvrz_EOMphi}
Let $\phi' = \phi + \epsilon \nu$, where $\nu \in C^{\infty}(P)$ is any scalar function vanishing on $\partial P$, and $\epsilon \in \R$. Then for fixed metric $g$ and $B \in \dfP{2}$, one has 
\begin{equation} \label{eq_Svarinphi}
S[g,B,\phi'] = S[g,B,\phi] - 2 \epsilon \int_{P} \{e^{-2\phi} \beta(\phi) \cdot \nu \}\cdot d\vol_{g} + o(\epsilon^{2}),
\end{equation}
where $\beta(\phi)$ is a scalar function defined by (\ref{def_Bphi}). 

Similarly, let $\phi'_{0} = \phi_{0} + \epsilon \nu_{0}$, where $\nu_{0} \in \cif$ is any scalar function vanishing on $\partial M$, and $\epsilon \in \R$. The for fixed metric $g_{0}$, $B_{0} \in \df{2}$ and $\vartheta \in \Omega^{1}(M,\g_{P})$, one has 
\begin{equation}
S_{0}[g_{0},B_{0},\phi'_{0},\vartheta] = S_{0}[g_{0},B_{0},\phi_{0},\vartheta] - 2 \epsilon \int_{M} \{ e^{-2\phi_{0}} \beta'(\phi_{0}) \} \cdot \nu_{0} \cdot d\vol_{g_{0}} + 
o(\epsilon^{2}),
\end{equation}
where $\beta'(\phi')$ is a scalar function defined by (\ref{def_B'phi0}). 
\end{tvrz}

\begin{proof}
We will prove only the first statement. Clearly $e^{-2\phi'} = e^{-2\phi}(1 - 2 \epsilon \nu) + o(\epsilon^{2})$. It is thus easy to see that we find
\begin{equation}
\begin{split}
S[g,B,\phi'] = S[g,B,\phi] & - 2 \epsilon \int_{P} e^{-2\phi} \{ \RS(g) - \frac{1}{2} \<H',H'\>_{g} + 4 \< d\phi, d\phi\>_{g} \} \cdot \nu \cdot d\vol_{g} \\
& + 2 \epsilon \int_{P} \<d\nu, e^{-2\phi} d\phi \>_{g} \cdot d\vol_{g} 
\end{split}
\end{equation}
We only have to treat to the last term. For this, we use the preceding lemma. Under the integral, one thus obtains the expression $\delta_{g}( e^{-2\phi} d\phi) \cdot \nu = e^{-2\phi} \{ 2 \<d\phi,d\phi\>_{g} - \Delta_{g}(\phi) \} \cdot \nu $. By plugging this back, we obtain precisely the expression (\ref{eq_Svarinphi}). 
\end{proof}
Using the equation (\ref{eq_dcodadjoint}), it is straightforward to find the equations of motion for the Kalb-Ramond fields. We formulate this as a proposition. 

\begin{tvrz} \label{tvrz_KRvariation}
Let $B' = B + \epsilon C$, where $C \in \dfP{2}$ is any $2$-form vanishing on $\partial P$, and $\epsilon \in \R$. Then for a fixed metric $g$ and $\phi \in C^{\infty}(P)$, one has 
\begin{equation}
S[g,B',\phi] = S[g,B,\phi] - 2 \epsilon \int_{M} e^{-2\phi} \< \beta(B), C \>_{g} \cdot d\vol_{g} + o(\epsilon^{2}),
\end{equation}
where $\beta(B) \in \dfP{2}$ is a $2$-form defined as $\beta(B) = \frac{1}{2} e^{2\phi} \delta_{g}( e^{-2\phi} H')$. 

Similarly, let $B'_{0} = B_{0} + \epsilon C_{0}$, where $C_{0} \in \Omega^{2}(P)$ is any $2$-form vanishing on $\partial M$, and $\epsilon \in \R$. Then for a fixed metric $g_{0}$, $\phi_{0} \in \cif$ and $\vartheta \in \Omega^{1}(M,\g_{P})$, one obtains 
\begin{equation}
S_{0}[g_{0},B'_{0},\phi_{0},\vartheta] = S_{0}[g_{0},B_{0},\phi_{0},\vartheta] - 2 \epsilon \int_{M} e^{-2\phi_{0}} \<\beta'(B_{0}),C_{0}\>_{g_{0}} \cdot d\vol_{g_{0}} + o(\epsilon^{2}),
\end{equation}
where $\beta'(B_{0}) \in \df{2}$ is a $2$-form defined as $\beta'(B_{0}) = \frac{1}{2} e^{2\phi_{0}} \delta_{g_{0}}(e^{-2\phi_{0}} H'_{0})$. 
\end{tvrz}
\begin{proof}
The proof follows from the fact that $H' = H + \epsilon \cdot dC$, combined with (\ref{eq_dcodadjoint}). The same argument applies for the equation regarding $S_{0}$. 
\end{proof}
To complete this section, we have to derive the variation of the action $S_{0}$ in the dynamical variable $\vartheta \in \Omega^{1}(M,\g_{P})$. To do so, first observe that such field can be used to define a new principal connection $A' \in \Omega^{1}(P,\g)$ on $P$. Let $j: \Gamma(\g_{P}) \rightarrow \mathfrak{X}_{G}(P)$ be an inclusion of sections of the adjoint bundle into the $\cif$-module of $G$-invariant vector fields on $P$. To each $X \in \vf{}$, define its horizontal lift $X'^{h} \in \mathfrak{X}_{G}(P)$ corresponding to $A'$ as 
\begin{equation} \label{eq_horliftchange}
X'^{h} = X^{h} - j(\vartheta(X)),
\end{equation}
where $X^{h}$ is a horizontal lift of $X$ corresponding to the original connection $A$. Let $D'$ be the exterior covariant derivative corresponding to $A'$. One finds its relation to $D$ in the form
\begin{equation} \label{eq_extcovprime}
D'\omega = D\omega + [\vartheta \^ \omega ]_{\g},
\end{equation}
for all $\omega \in \Omega^{p}(M,\g_{P})$. Let $F' \in \Omega^{2}(M,\g_{P})$ be the curvature $2$-form corresponding to $A'$. Then
\begin{equation} \label{eq_F'asF}
F' = F + D\vartheta + \frac{1}{2}[\vartheta \^ \vartheta]_{\g}. 
\end{equation}
But according to (\ref{eq_F'andH'0}), this is exactly the same $2$-form as $F'$ in the action $S_{0}$. 
\begin{tvrz} \label{tvrz_EOMtheta}
Let $\vartheta' = \vartheta + \epsilon \eta$, where $\eta \in \Omega^{1}(M,\g_{P})$ is any $2$-form vanishing on $\partial M$, and $\epsilon \in \R$. Then for a fixed metric $g_{0}$, $B_{0} \in \df{2}$ and $\phi_{0} \in \cif$, one has 
\begin{equation} \label{eq_varintheta}
S_{0}[g_{0},B_{0},\phi_{0},\vartheta'] = S[g_{0},B_{0},\phi_{0},\vartheta] + 2 \epsilon \int_{M} e^{-2\phi_{0}} \{ \dal \beta'_{\vartheta}, \eta \dar - \frac{1}{2} \< \beta'_{B_{0}}, \vartheta \^ \eta \>_{\g} \>_{g} \} \cdot d\vol_{g_{0}},
\end{equation}
where $\beta'_{\vartheta} \in \Omega^{1}(M,\g_{P})$ is a $\g_{P}$-valued $1$-form on $M$ defined by 
\begin{equation}
\beta'_{\vartheta} = \frac{1}{2} \{ e^{2\phi_{0}} D'^{\dagger}(e^{-2\phi_{0}} F') + \< \io_{e_{k}}H'_{0}, F'\>_{g_{0}} \cdot e^{k} \},
\end{equation}
where $\{ e_{k} \}_{k=1}^{\dim{M}}$ is any local frame on $M$. Here $D'^{\dagger}$ is a covariant codiferential corresponding to the exterior covariant derivative $D'$ defined by (\ref{eq_extcovprime}). $\beta_{B_{0}} \in \df{2}$ is defined in Proposition \ref{tvrz_KRvariation}. 
\end{tvrz}
\begin{proof}
One can start with the variation of the kinetic term 
\begin{equation} S_{kin}[g_{0},\phi_{0},\vartheta] = \frac{1}{2} \int_{M} e^{-2 \phi_{0}}  \dal F', F' \dar \cdot d\vol_{g_{0}}. \end{equation}
This is easy using (\ref{eq_extcovprime}, \ref{eq_F'asF}) and (\ref{eq_Dcodadjoint}). Indeed, under the variation $\vartheta \mapsto \vartheta + \epsilon \eta$, we gets $F' \mapsto F' + \epsilon D'(\eta) + o(\epsilon^{2})$, and the action thus changes as 
\begin{equation} \label{eq_S0kinvar}
S_{kin}[g_{0},\phi_{0},\vartheta'] = S_{kin}[g_{0},\phi_{0},\vartheta] + 2 \epsilon \int_{M} e^{-2\phi_{0}} \dal \frac{1}{2} e^{2\phi_{0}} D'^{\dagger}( e^{-2\phi_{0}}F'), \eta \dar \cdot d\vol_{g_{0}} + o(\epsilon^{2}). 
\end{equation}
The change of $H'_{0}$ takes the explicit form
\begin{equation}
H'_{0} \mapsto H'_{0} + \frac{1}{2} \epsilon  \{  \<D'(\vartheta) \^ \eta\>_{\g} - \< \vartheta  \^ D'(\eta) \>_{\g} \} - \epsilon \< F' \^ \eta \>_{\g} + o(\epsilon^{2}).
\end{equation}
Now, note that for any $\alpha, \beta \in \Omega^{\bullet}(M,\g_{P})$, there holds an equation
\begin{equation}
d \< \alpha \^ \beta \>_{\g} = \< D(\alpha) \^ \beta \>_{\g} + (-1)^{|\alpha|} \<\alpha \^ D(\beta) \>_{\g}. 
\end{equation}
We find that $H'_{0} \mapsto H'_{0} + \epsilon \cdot \{ \frac{1}{2}  d \< \vartheta \^ \eta \>_{\g} - \<F' \^ \eta  \>_{\g} \} + o(\epsilon^{2})$. Let $S_{H'_{0}}$ be the part of the action quadratic in $H'_{0}$. One obtains a relation 
\begin{equation}
\begin{split}
S_{H'_{0}}[ g_{0},B_{0},\phi_{0},\vartheta'] = & \ S'_{H'_{0}}[g_{0},B_{0},\phi_{0},\vartheta] \\
 + \epsilon \int_{M} & e^{-2\phi_{0}} \{ \< \<F' \^ \eta\>_{\g} - \frac{1}{2} d\<\vartheta \^ \eta\>_{\g}, H'_{0} \>_{g_{0}} \} \cdot d\vol_{g_{0}} + o(\epsilon^{2}).
 \end{split}
\end{equation}
To finish the proof, observe that there holds the equation $\< \<F' \^ \eta \>_{\g}, H'_{0} \>_{g_{0}} = \dal \< \io_{e_{k}}H'_{0}, F'\>_{g_{0}} \cdot e^{k}, \eta \dar$ which can be proved by a direct calculation. Using (\ref{eq_dcodadjoint}), we find 
\begin{equation}
\begin{split}
S_{H'_{0}}[g_{0},B_{0},\phi_{0},\vartheta'] = & \  S_{H'_{0}}[g_{0},B_{0},\phi_{0},\vartheta] 
 + 2 \epsilon \int_{M} e^{-2\phi_{0}} \dal \frac{1}{2} \< \io_{e_{k}}H'_{0}, F'\>_{\g} \cdot \psi^{k}, \eta \dar \cdot d\vol_{g_{0}} \\
& -  \epsilon \int_{M} e^{-2 \phi_{0}} \< \frac{1}{2} e^{2\phi_{0}} \delta_{g}( e^{-2\phi_{0}} H'_{0}), \< \vartheta \^ \eta \>_{\g} \>_{g_{0}} \cdot d\vol_{g_{0}} + o(\epsilon^{2}) 
\end{split}
\end{equation}
Combining this with (\ref{eq_S0kinvar}) gives exactly the equation (\ref{eq_varintheta}). 
\end{proof}
\begin{rem}
The covariant codifferential $D^{\dagger}$ can be conveniently expressed using the covariant derivative $D_{X}$ corresponding to $A$, and the Levi-Civita connection $\cD^{0}$ corresponding to $g_{0}$. Indeed, for any $\omega \in \Omega^{p}(M,\g_{P})$, define a combined covariant derivative $\hcD^{0}_{X}$ as 
\begin{equation}
\begin{split}
\{ \hcD^{0}_{X}\omega \}(X_{1},\dots,X_{p}) = & \ D_{X}\{ \omega(X_{1},\dots,X_{p}) \} \\
& - \omega(\cD^{0}_{X}(X_{1}), \dots, X_{p}) \cdots - \omega(X_{1},\dots,\cD^{0}_{X}(X_{p})),
\end{split}
\end{equation}
for all $X, X_{1}, \dots, X_{p} \in \vf{}$. The covariant codifferential can be then expressed as 
\begin{equation} \label{eq_covariantdifexpl}
\{ D^{\dagger}(\omega) \}(X_{1},\dots,X_{p-1}) = - (\hcD^{0}_{e_{k}} \omega)(g_{0}^{-1}(\psi^{k}),X_{1},\dots,X_{p-1}),
\end{equation}
for all $X_{1},\dots,X_{p-1} \in \vf{}$. Here $\{ e_{k} \}_{k=1}^{\dim{M}}$ is any local frame on $M$. 
\end{rem}
\begin{rem}
The $1$-form $\< \io_{e_{k}} H'_{0}, F' \>_{g_{0}} \cdot e^{k}$ can be rewritten without the explicit appearance of the local frame, but with some additional signs:
\begin{equation}
\< \io_{e_{k}} H'_{0}, F' \>_{g_{0}} \cdot e^{k} = sgn(g_{0}) (-1)^{n+1} \ast( F' \^ \ast H'_{0}),
\end{equation}
where $n = \dim{M}$ and $sgn(g_{0})$ is the signature of the determinant of $g_{0}$. 
\end{rem}

We can now propose the final theorem of this section, summarizing the already proved statements. For explicit forms of all kinds of beta functions see the respective propositions. 
\begin{theorem}[\textbf{Equations of motion}]\label{thm_EOM}
The fields $(g,B,\phi)$ satisfy the equations of motion given by the action (\ref{eq_S}), if and only if 
\begin{equation} \label{eq_thmEOM1} 
\beta_{g} = \beta_{B} = \beta_{\phi} = 0. 
\end{equation}
Similarly, the fields $(g_{0},B_{0},\phi_{0},\vartheta)$ satisfy the equations of motion given by (\ref{eq_S0}), if and only if 
\begin{equation} \label{eq_thmEOM2}
\beta'_{g_{0}} = \beta'_{B_{0}} = \beta'_{\phi_{0}} = \beta'_{\vartheta} = 0. 
\end{equation}
Involved tensor fields are defined in the propositions \ref{tvrz_EOMg}, \ref{tvrz_EOMg0}, \ref{tvrz_KRvariation} and \ref{tvrz_EOMtheta}.
\end{theorem}
\section{Courant algebroids and Levi-Civita connections} \label{sec_Courant}
In this section, we will very briefly recall the geometrical objects required to describe the equations of motion presented in the previous section. For definitions with full details, examples and proved propositions, see our lecture notes \cite{Jurco:2016emw}. We will present the results mostly without the detailed calculations, as they are not very enlightening, yet requiring a lot of space to be written in all details. 

\begin{definice}
Let $q: E \rightarrow M$ be a vector bundle, equipped with a fiber-wise metric $\<\cdot,\cdot\>_{E}$, a vector bundle morphism $\rho: E \rightarrow M$, and an $\R$-bilinear bracket $[\cdot,\cdot]_{E}$, such that 
\begin{align}
\label{eq_leibnizrule} [\psi,f\psi']_{E} & = f[\psi,\psi']_{E} + (\rho(\psi).f) \psi', \\
\label{eq_leibnizidentity} [\psi,[\psi',\psi'']_{E}]_{E} & = [[\psi,\psi']_{E},\psi'']_{E} + [\psi', [\psi,\psi'']_{E}]_{E}, \\
\label{eq_killingequation} \rho(\psi).\<\psi',\psi''\>_{E} & = \<[\psi,\psi']_{E},\psi''\>_{E} + \<\psi', [\psi,\psi'']_{E}\>_{E}, \\
\label{eq_sympart} [\psi,\psi']_{E} & = -[\psi',\psi]_{E} + \D\<\psi,\psi'\>_{E},
\end{align}
for all $\psi,\psi',\psi'' \in \Gamma(E)$ and $f \in \cif$. Here $\D: \cif \rightarrow \Gamma(E)$ is defined as 
\begin{equation} \label{def_Dmap}
\< \D{f}, \psi \>_{E} = \rho(\psi).f, 
\end{equation}
for all $f \in \cif$ and $\psi \in \Gamma(E)$. Then $(E,\rho,\<\cdot,\cdot\>_{E},[\cdot,\cdot]_{E})$ is called the \textbf{Courant algebroid}. Note that $\rho$ is automatically a bracket morphism, that is 
\begin{equation} \label{eq_rhohom}
\rho([\psi,\psi']_{E}) = [\rho(\psi),\rho(\psi')],
\end{equation}
for all $\psi,\psi' \in \Gamma(E)$. We sometimes write $g_{E} = \<\cdot,\cdot\>_{E}$ to avoid the bracket notation. 
\end{definice}
\begin{example} \label{ex_dorfman}
Let $E = \gTM \equiv TM \oplus T^{\ast}M$, and let $\rho = \pi_{1} \in \Hom(\gTM,TM)$ be the projection. Let $\<\cdot,\cdot\>_{E}$ be the canonical pairing of $TM$ and $T^{\ast}M$, and let 
\begin{equation}
[(X,\xi),(Y,\eta)]_{E} = ( [X,Y], \Li{X}\eta - \io_{Y}d\xi - H(X,Y,\cdot)),
\end{equation}
for all $(X,\xi),(Y,\eta) \in \Gamma(\gTM)$, where $H \in \df{3}$ is a closed $3$-form on $M$. Then the $4$-tuple $(\gTM,\rho,\<\cdot,\cdot\>_{E},[\cdot,\cdot]_{E})$ forms a Courant algebroid called the \textbf{$H$-twisted Dorfman bracket}. If $H$ is not closed, one obtains a structure of pre-Courant algebroid, see \cite{2012arXiv1205.5898L}. For applications in this paper, this is sufficient, as we never have to use the full Leibniz identity (\ref{eq_leibnizidentity}). 
\end{example}

Next, let us recall the notion of a generalized Riemannian metric. In essence, we consider certain positive definite fiber-wise metrics on $E$, compatible with the already present pairing $\<\cdot,\cdot\>_{E}$. There are many reformulations of this concept which we recall without proofs. 
\begin{definice}
Let $\gm$ be a positive definite fiber-wise metric on $(E,\<\cdot,\cdot\>_{E})$ We say that $\gm$ is a \textbf{generalized Riemannian metric on $E$} if the induced map $\gm \in \Hom(E,E^{\ast})$ is orthogonal with respect to the fiber-wise metric $\<\cdot,\cdot\>_{E}$ on $E$ and the standard induced one $\<\cdot,\cdot\>_{E^{\ast}}$ on the dual vector bundle $E^{\ast}$. Equivalently, $\gm$ defines an orthogonal and symmetric involution $\tau \in \End(E)$ related to $\gm$ via $\gm(\psi,\psi') = \<\psi,\tau(\psi')\>_{E}$, for all $\psi, \psi' \in \Gamma(E)$. Finally, $\gm$ defines a maximal positive definite subbundle $V_{+} \subset E$ with respect to $\<\cdot,\cdot\>_{E}$. In particular, one has 
\begin{equation}
E = V_{+} \oplus V_{-},
\end{equation}
where $V_{-} = V_{+}^{\perp}$ forms a maximal negative definite subbundle of $E$. $V_{\pm}$ are obtained as $\pm 1$ eigenbundles of the involution $\tau$. 
\end{definice}
\begin{example}
Let $E = \gTM$ and $\<\cdot,\cdot\>_{E}$ as above. Then the most general generalized Riemannian metric $\gm$ takes the block form 
\begin{equation} \label{eq_exgenmetric}
\gm = \bm{g - Bg^{-1}B}{Bg^{-1}}{-g^{-1}B}{g^{-1}}, 
\end{equation}
for a Riemannian metric $g$ on $M$ and $B \in \df{2}$. For applications in physics, $g$ is usually not Riemannian. However, this poses no issues, as $\gm$ defined by formula above remains a fiber-wise metric (of indefinite signature) on $\gTM$. To avoid any such discussions, we will henceforth assume that $g > 0$. 
\end{example}
Next, we can recall a concept of Courant algebroid connections. In substance, they generalize both vector bundle connections and usual manifold linear connections. In particular, we assume its compatibility with the already present fiber-wise metric $\<\cdot,\cdot\>_{E}$. Note that the definition itself does not use the Courant bracket $[\cdot,\cdot]_{E}$. 

\begin{definice}
Let $(E,\rho,\<\cdot,\cdot\>_{E},[\cdot,\cdot]_{E})$ be a Courant algebroid. An $\R$-linear map $\cD: \Gamma(E) \times \Gamma(E) \rightarrow \Gamma(E)$ is called a Courant algebroid connection, if the operator $\cD_{\psi} \equiv \cD(\psi,\cdot)$ satisfies
\begin{equation} \label{eq_conleibniz}
\cD_{f\psi} \psi' = f \cD_{\psi} \psi', \; \; \cD_{\psi}(f\psi') = f \cD_{\psi}\psi' + (\rho(\psi).f) \psi',
\end{equation}
for all $\psi,\psi' \in \Gamma(E)$, $f \in \cif$, and 
\begin{equation}
\rho(\psi).\<\psi',\psi''\>_{E} = \< \cD_{\psi}\psi', \psi''\>_{E} + \<\psi', \cD_{\psi}\psi''\>_{E}, 
\end{equation}
for all $\psi,\psi',\psi'' \in \Gamma(E)$. Equivalently, one can write this as $\cD{g_{E}} = 0$, where $\cD: \T_{p}^{q}(E) \rightarrow \T_{p+1}^{q}(E)$ is the naturally induced covariant differential on the tensor algebra of $E$. 
\end{definice}
\begin{example}
For any vector bundle $q: E \rightarrow M$ with a fiber-wise metric $\<\cdot,\cdot\>_{E}$, one can always construct an ordinary vector bundle connection $\cD': \vf{} \times \Gamma(E) \rightarrow \Gamma(E)$ which satisfies $\cD' g_{E} = 0$. Setting $\cD_{\psi} \psi' := \cD'_{\rho(\psi)} \psi'$ for all $\psi,\psi' \in \Gamma(E)$ defines a Courant algebroid connection. 
\end{example}

Having connections defined in a balanced way (both inputs are sections of $E$), it makes sense to define a torsion operator. However, this has some difficulties which have been overcome independently in \cite{2007arXiv0710.2719G} and \cite{alekseevxu}. 
\begin{definice}
Let $\cD$ be a Courant algebroid connection. Define an $\R$-trilinear map $T_{G}$ as 
\begin{equation}
T_{G}(\psi,\psi',\psi'') = \< \cD_{\psi}\psi' - \cD_{\psi'}\psi - [\psi,\psi']_{E}, \psi''\>_{E} + \< \cD_{\psi''}\psi, \psi'\>_{E},
\end{equation}
for all $\psi,\psi',\psi'' \in \Gamma(E)$. The map $T_{G}$ is $\cif$-trilinear and skew-symmetric, that is $T_{G} \in \Omega^{3}(E)$, and called a \textbf{torsion $3$-form of $\cD$}. It is related to the \textbf{torsion operator $T$ of $\cD$} as $\< T(\psi,\psi'), \psi''\>_{E} := T_{G}(\psi,\psi',\psi'')$. The connection is called \textbf{torsion-free}, if $T_{G} = 0$. 
\end{definice}
One can also attempt to define a Riemann curvature tensor of $\cD$. This is even more intriguing. Our intention was to obtain a tensorial quantity with enough symmetries to unambiguously define a Ricci tensor. To achieve this, we have taken an inspiration in physics, namely in double field theory and \cite{Hohm:2012mf}. Note that it is well defined for any Courant algebroid and connection without the requirement of any additional structures. 
\begin{definice}
Let $\cD$ be a Courant algebroid connection. Define a map $R$ as 
\begin{equation} \label{def_Rtensor}
R(\phi',\phi,\psi,\psi') = \frac{1}{2} \{ R^{(0)}(\phi',\phi,\psi,\psi') + R^{(0)}(\psi',\psi,\phi,\phi') + \< \cD_{\psi_{\lambda}} \psi, \psi'\>_{E} \cdot \< \cD_{\psi^{\lambda}_{E}} \phi, \phi'\>_{E} \},
\end{equation}
where $R^{(0)}$ is a \textit{naive} curvature operator defined by 
\begin{equation}
R^{(0)}(\phi',\phi,\psi,\psi') = \< \phi', \{ [\cD_{\psi},\cD_{\psi'}] - \cD_{[\psi,\psi']_{E}} \} \phi \>_{E},
\end{equation}
for all $\phi',\phi,\psi,\psi' \in \Gamma(E)$. Here $\{ \psi_{\lambda} \}_{\lambda=1}^{\rank(E)}$ is an arbitrary local frame on $E$ and $\psi^{\lambda}_{E} \equiv g_{E}^{-1}(\psi^{\lambda})$ is a frame satisfying $\< \psi_{\lambda}, \psi^{\kappa}_{E} \>_{E} = \delta_{\lambda}^{\kappa}$. Then $R$ is $\cif$-linear in all inputs, $R \in \T_{4}^{0}(E)$, and called a \textbf{Riemann curvature tensor corresponding to $\cD$}. The only non-trivial (up to a sign) partial trace of $R$ is called a \textbf{Ricci curvature tensor} and defined as 
\begin{equation} \label{def_Rictensor}
\Ric(\psi,\psi') = R( \psi^{\lambda}_{E}, \psi', \psi_{\lambda}, \psi), 
\end{equation}
for all $\psi,\psi' \in \Gamma(E)$. This tensor is symmetric, see \cite{Jurco:2016emw}. Finally, one can define a smooth function 
\begin{equation}
\RS_{E} = \Ric(\psi^{\lambda}_{E},\psi_{\lambda}), 
\end{equation}
called a \textbf{Courant-Ricci scalar curvature of $\cD$}. 
\end{definice}
In the written text, we will usually omit some adjectives describing $R$, $\Ric$ and $\RS_{E}$. Note that all what is said is true also for pre-Courant algebroids. However, the bracket must satisfy (\ref{eq_rhohom}) and all Courant algebroid and connection axioms with the exception of Leibniz identity in order to prove the symmetries of the tensor $R$. 
\begin{definice}
Let $\cD$ be a Courant algebroid connection. A \textbf{covariant divergence corresponding to $\cD$} is an $\R$-linear map $\Div_{\cD}: \Gamma(E) \rightarrow \cif$ defined by
\begin{equation} \label{eq_Divoperator}
\Div_{\cD}(\psi) = \< \cD_{\psi_{\lambda}} \psi, \psi^{\lambda}_{E} \>_{E},
\end{equation}
for all $\psi \in \Gamma(E)$, where $\{ \psi_{\lambda} \}_{\lambda=1}^{\rank(E)}$ is any local frame on $E$. This map satisfies a Leibniz rule:
\begin{equation}
\Div_{\cD}(f\psi) = f \Div_{\cD}(\psi) + \rho(\psi).f
\end{equation}
One can obtain a \textbf{characteristic vector field of $\cD$} as $X_{\cD} \in \vf{}$ defined for $f \in \cif$ as 
\begin{equation}
X_{\cD}.f = \Div_{\cD}(\D{f}),
\end{equation}
where $\D: \cif \rightarrow \Gamma(E)$ is the map (\ref{def_Dmap}). 
\end{definice}
\begin{definice}
Let $(E,\rho,\<\cdot,\cdot\>_{E},[\cdot,\cdot]_{E})$ be a Courant algebroid equipped with a generalized metric $\gm$. Let $\cD$ be a Courant algebroid connection on $E$. We say that $\cD$ is a \textbf{Levi-Civita connection on $E$ with respect to $\gm$}, if $\cD \gm = 0$ and $\cD$ is torsion-free. 
\end{definice}
Recalling the ordinary Riemannian geometry, one may attempt to use the same procedure to obtain a closed formula for a Levi-Civita connection. However, this is not possible as one quickly finds out that there are infinitely many Levi-Civita connections. For the proof of the existence of Levi-Civita connections see \cite{Garcia-Fernandez:2016ofz}. For exact Courant algebroids, there is a full classification in \cite{Jurco:2016emw} or \cite{Jurco:2015bfs}. There is one remarkable property of the introduced structures. All of them transform in a covariant way under Courant algebroid isomorphisms. In particular, the characteristic vector field forms an invariant. See \cite{Jurco:2016emw} for details. Having a generalized metric $\gm$, one can introduce the following notions:
\begin{definice}
Let $\cD$ be a Courant algebroid connection on $E$ equipped with a generalized metric $\gm$. One says that $\cD$ is \textbf{Ricci compatible with $\gm$} if $\Ric(V_{+},V_{-}) = 0$. Moreover, define a Ricci scalar curvature $\RS_{\gm}$ corresponding to $\gm$ as 
\begin{equation}
\RS_{\gm} = \Ric( \gm^{-1}(\psi^{\lambda}), \psi_{\lambda}),
\end{equation}
where $\{ \psi_{\lambda} \}_{\lambda=1}^{\rank(E)}$ is any local frame on vector bundle $E$. 
\end{definice}
\section{Equations of motion in terms of connections} \label{sec_EOMgeom}
It is a remarkable fact that both systems of equations of motion (\ref{eq_thmEOM1}) and (\ref{eq_thmEOM2}) can be geometrically described in terms of Levi-Civita connections on Courant algebroids. Note that similar approach was taken in \cite{2013arXiv1304.4294G} using a slightly different language. For system (\ref{eq_thmEOM1}), we have decided not to include a full calculation here, as it is explicitly calculated in \cite{Jurco:2016emw}. We will examine all details of the similar statement for the system (\ref{eq_thmEOM2}). 
\begin{theorem} \label{thm_eomSG}
Let $E = \gTP$ be equipped with the Courant algebroid structure described in Example \ref{ex_dorfman}. Let $\gm$ be a generalized metric (\ref{eq_exgenmetric}) corresponding to a pair $(g,B)$. Let $\cD$ be a Levi-Civita connection on $\gTP$ with respect to $\gm$, such that $X_{\cD} = 0$. Moreover, assume that there is a smooth function $\phi \in C^{\infty}(P)$, such that 
\begin{equation}
(d\phi)(Z) = \< \cD_{\rho^{\ast}(e^{k})} \rho^{\ast}(g(e_{k})), \rho^{\ast}(g(Z)) \>,
\end{equation}
where $\rho^{\ast} \in \Hom(T^{\ast}P,E)$ is a map defined as $\rho^{\ast} = g_{E}^{-1} \circ \rho^{T}$, and $\{ e_{k} \}_{k=1}^{\dim{M}}$ is some local frame.

Then $(g,B,\phi)$ satisfy the equations of motion (\ref{eq_thmEOM1}), if and only if $\RS_{\gm} = 2 \Lambda$ and $\cD$ is Ricci compatible with $\gm$, that is $\Ric(V_{+},V_{-}) = 0$. The Ricci compatibility is equivalent to $\beta_{g} = \beta_{B} = 0$.
Moreover, one has $\RS_{E} = 0$. 
\end{theorem}
\subsection{Suitable geometry}
To find a similar statement for the system (\ref{eq_thmEOM2}), one let us first introduce a Courant algebroid structure on $E' = TM \oplus \g_{P} \oplus T^{\ast}M$. Let $\rho' = \pi_{1}$, a projection on the first factor of the direct sum. The fiber-wise metric $\<\cdot,\cdot\>_{E'}$ is defined as 
\begin{equation}
\< (X,\Phi,\xi), (Y,\Phi',\eta) \>_{E'} = \eta(X) + \xi(Y) + \< \Phi,\Phi'\>_{\g},
\end{equation}
for all $(X,\Phi,\xi)$, $(Y,\Phi',\eta) \in \Gamma(E')$. Finally, the bracket takes the form
\begin{equation} \label{eq_E'bracket}
\begin{split}
[(X,\Phi,\xi),(Y,\Phi',\eta)]_{E'} = \big( & [X,Y], D_{X}\Phi' - D_{Y}\Phi' - F(X,Y) - [\Phi,\Phi']_{\g}, \Li{X}\eta - \io_{Y}d\xi \\
& - H_{0}(X,Y,\cdot) - \<F(X),\Phi'\>_{\g} + \<F(Y),\Phi\>_{\g} + \< D \Phi, \Phi'\>_{\g} \big),
\end{split}
\end{equation}
where $D$ is the covariant exterior derivative corresponding to a fixed principal bundle connection $A \in \Omega^{1}(P,\g)$, and $F \in \Omega^{2}(M,\g_{P})$ is the corresponding curvature $2$-form. Observe that in first two components $[\cdot,\cdot]_{E'}$ coincides with the Atiyah Lie algebroid corresponding to $P$. Then $(E',\rho',\<\cdot,\cdot\>_{E'},[\cdot,\cdot]_{E'})$ forms a Courant algebroid, if and only if there holds an equation
\begin{equation} \label{eq_pontryiaginvanishes}
dH_{0} + \frac{1}{2} \< F \^ F \>_{\g} = 0.
\end{equation}
In particular, the first Pontriyagin class of the principal bundle $P$ must vanish. However, we may consider the case of a general $P$, obtaining a pre-Courant algebroid instead. Every generalized metric $\gm'$ on $E'$ can be uniquely parametrized by a triple $(g_{0},B_{0},\vartheta)$ and written as 
\begin{equation} \label{eq_G'metric}
\gm' = 
\begin{pmatrix}
1 & \vartheta^{T} & B_{0} - \frac{1}{2} \vartheta c \vartheta^{T} \\
0 & 1 & -c \vartheta \\
0 & 0 & 1 
\end{pmatrix}
\begin{pmatrix}
g_{0} & 0 & 0 \\
0 & -c & 0 \\
0 & 0 & g_{0}^{-1}
\end{pmatrix}
\begin{pmatrix}
1 & 0 & 0 \\
\vartheta & 1 & 0 \\
-B_{0} - \frac{1}{2} \vartheta^{T} c \vartheta & - \vartheta^{T}c & 1 
\end{pmatrix},
\end{equation}
where $\vartheta \in \Omega^{1}(M,\g_{P})$ is viewed as a vector bundle map $\vartheta \in \Hom(TM,\g_{P})$. This is true only for compact $\g$. For general semisimple Lie algebra, the positive definite fiber-wise metric $-c$ in the middle block can be more general. $\gm'$ can be written as $\gm' = (e^{-\C})^{T} \G' e^{-\C}$, where 
\begin{equation} \label{eq_defG'andC}
\G' = \begin{pmatrix}
g_{0} & 0 & 0 \\
0 & -c & 0 \\
0 & 0 & g_{0}^{-1}
\end{pmatrix}, \; \; 
\C = \begin{pmatrix}
0 & 0 & 0 \\
-\vartheta & 0 & 0 \\
B_{0} & \vartheta^{T}c & 0 
\end{pmatrix}.
\end{equation}
Instead of working with the above Courant algebroid and generalized metric $\gm'$, it is convenient to consider the following twisted structure. Define the bracket $[\cdot,\cdot]'_{E'}$ as 
\begin{equation}
[\psi,\psi']'_{E'} = e^{-\C}( [e^{\C}(\psi), e^{\C}(\psi') ]_{E'},
\end{equation}
for all $\psi,\psi' \in \Gamma(E')$. As $e^{\C}$ is orthogonal with respect to $\<\cdot,\cdot\>_{E'}$ and $\rho' \circ e^{\C} = \rho'$, we get that $(E',\rho',\<\cdot,\cdot\>_{E'}, [\cdot,\cdot]'_{E'})$ is again a (pre-)Courant algebroid. Let $\cD'$ be a Levi-Civita connection on $E'$ with respect to $\gm$, using the bracket (\ref{eq_E'bracket}). We define a new connection $\hcD'$ as 
\begin{equation} \label{eq_hcD'twist}
\hcD'_{\psi}\psi' = e^{-\C}( \cD'_{e^{\C}(\psi)} e^{\C}(\psi')).
\end{equation}
It is easy to see that $\hcD$ forms a Levi-Civita connection on $E'$ with respect to the block diagonal generalized metric $\G'$, using the twisted bracket $[\cdot,\cdot]'_{E'}$. In fact, this bracket takes the same form as (\ref{eq_E'bracket}), but with primes added to all of the quantities. Explicitly: 
\begin{equation} \label{eq_E'bracket'}
\begin{split}
[(X,\Phi,\xi),(Y,\Phi',\eta)]'_{E'} = \big( & [X,Y], D'_{X}\Phi' - D'_{Y}\Phi - F'(X,Y) - [\Phi,\Phi']_{\g}, \Li{X}\eta - \io_{Y}d\xi \\
& - H'_{0}(X,Y,\cdot) - \<F'(X),\Phi'\>_{\g} + \<F'(Y),\Phi\>_{\g} + \< D' \Phi, \Phi'\>_{\g} \big),
\end{split}
\end{equation}
where $H'_{0}$ and $F'$ are defined by (\ref{eq_F'andH'0}) and $D' = D + [\vartheta \^ \cdot]$. This is not surprising, as the twist using $\C$ in fact corresponds to the choice of a connection $A'$ instead of $A$, see (\ref{eq_horliftchange}), combined with the twist by the $2$-form $B_{0}$. 

\subsection{The connection and its curvatures}
We will now construct an example of Levi-Civita connection on $E'$ with respect to $\G'$ and calculate its scalar curvature and equations equivalent to its Ricci compatibility with $\G'$. Set 
\begin{align}
\label{eq_hcD'1} \hcD'_{(X,0,0)} & = \begin{pmatrix}
\cDN_{X} & \frac{1}{2} g_{0}^{-1} \<F'(X),\star\>_{\g} & -\frac{1}{3} g_{0}^{-1} H'_{0}(X,g_{0}^{-1}(\star), \cdot) \\
-\frac{1}{2} F'(X,\star) & D'_{X} & \frac{1}{2} F'(X,g_{0}^{-1}(\star)) \\
- \frac{1}{3} H'_{0}(X,\star,\cdot) & - \frac{1}{2} \< F'(X),\star\>_{\g} & \cDN_{X}
\end{pmatrix}, \\
\label{eq_hcD'2} \hcD'_{(0,\Phi,0)} & = \begin{pmatrix}
\frac{1}{2} g_{0}^{-1} \<F'(\star),\Phi\>_{\g} & 0 & 0 \\
0 & -\frac{1}{3} [\Phi,\star]_{\g} & 0 \\
0 & 0 & \frac{1}{2} \<F'(g_{0}^{-1}(\star)),\Phi\>_{\g}
\end{pmatrix}, \\
\label{eq_hcD'3} \hcD'_{(0,0,\xi)} & = \begin{pmatrix}
\frac{1}{6} g_{0}^{-1}H'_{0}(g_{0}^{-1}(\xi),\star,\cdot) & 0 & 0 \\
0 & 0 & 0 \\
0 & 0 & \frac{1}{6} H'_{0}(g_{0}^{-1}(\xi),g_{0}^{-1}(\star),\cdot),
\end{pmatrix}
\end{align}
where $\star$ always indicates the input. It is a straightforward check that $\hcD'$ is indeed a Levi-Civita connection on $E'$ with respect to $\G'$. It is now convenient to write
\begin{equation} \label{eq_hcD'decomp}
\hcD'_{\psi} \psi' = \hcDL_{\psi}\psi' + g_{E'}^{-1} \H(\psi,\psi',\cdot),
\end{equation}
where $\hcDL$ is a block diagonal induced connection $\hcDL_{(X,\Phi,\xi)} = \BlockDiag( \cDN_{X}, D'_{X}, \cDN_{X})$, and $\H \in \Omega^{1}(E') \otimes \Omega^{2}(E')$ takes the explicit form 
\begin{equation} \label{eq_Htensor}
\begin{split}
\H(\psi,\psi',\psi'') =& \ \frac{1}{6} H'_{0}(g_{0}^{-1}(\xi),Y,g_{0}^{-1}(\zeta)) + \frac{1}{6} H'_{0}(g_{0}^{-1}(\xi),g_{0}^{-1}(\eta),Z) \\
& - \frac{1}{3} H'_{0}(X,Y,Z) - \frac{1}{3} H'_{0}(X,g_{0}^{-1}(\eta),g_{0}^{-1}(\zeta))  - \frac{1}{3} \< [\Phi,\Phi']_{\g}, \Phi''\>_{\g} \\
& + \frac{1}{2} \< F'(g_{0}^{-1}(\zeta) - Z, X), \Phi'\>_{\g} - \frac{1}{2} \<F'(g_{0}^{-1}(\eta) - Y,X), \Phi''\>_{\g} \\
& + \frac{1}{2} \<F'(g_{0}^{-1}(\zeta),Y) - F'(g_{0}^{-1}(\eta),Z),\Phi \>_{\g}, \\
\end{split}
\end{equation}
where $\psi = (X,\Phi,\xi)$, $\psi' = (Y,\Phi',\eta)$ and $\psi'' = (Z,\Phi'',\zeta)$. 

Note that the connection $\hcDL$ is metric compatible with the generalized metric $\G'$, but it is not torsion free. Explicitly, its torsion $3$-form $\widehat{T}^{LC}_{G}$ reads
\begin{equation}
\begin{split}
\widehat{T}^{LC}_{G}((X,\Phi,\xi),(Y,\Phi',\eta),(Z,\Phi'',\zeta)) = & \ H'_{0}(X,Y,Z) + \<[\Phi,\Phi']_{\g},\Phi''\>_{\g} + \<F'(X,Y), \Phi''\>_{\g} \\
& + \<F'(Z,X), \Phi'\>_{\g} + \<F'(Y,Z), \Phi\>_{\g}
\end{split}
\end{equation}
However, the connection $\hcD'$ is a Courant algebroid connection on $E'$ both metric compatible with $\G'$ and torsion-free, which reflects in the following properties of $\H$:
\begin{align}
\label{eq_H1} \H(\psi,\psi',\psi'') + \H(\psi,\psi'',\psi') & = 0, \\
\label{eq_H2} \H(\psi,\psi',\tau'_{0}(\psi'')) + \H(\psi,\psi'',\tau'_{0}(\psi')) &= 0, \\
\label{eq_H3} \H(\psi,\psi',\psi'') + cyclic(\psi,\psi',\psi'') & = -\widehat{T}^{LC}_{G}(\psi,\psi',\psi''),
\end{align}
for all $\psi,\psi',\psi'' \in \Gamma(E')$. Here $\tau'_{0} \in \End(E')$ is the involution corresponding to the generalized metric $\G'$. We will now derive a formula for the Ricci tensor $\Ric'$ corresponding to $\hcD'$. Plugging (\ref{eq_hcD'decomp}) into (\ref{def_Rtensor}, \ref{def_Rictensor}), it is straightforward to arrive to the formula:
\begin{equation} \label{eq_Ric'expansioninH}
\begin{split}
\Ric'(\psi,\psi') = & \ \hRic{}^{LC}(\psi,\psi') + \frac{1}{2} \{ (\hcDL_{\psi_{\lambda}} \H)(\psi,\psi',\psi^{\lambda}_{E}) + (\hcDL_{\psi_{\lambda}} \H)(\psi',\psi,\psi^{\lambda}_{E}) \\
& - \H(\psi', g_{E}^{-1}\H(\psi_{\lambda},\psi,\cdot), \psi^{\lambda}_{E}) - \H(\psi^{\lambda}_{E}, g_{E}^{-1}\H(\psi,\psi_{\lambda},\cdot),\psi') \\
& + \H( g_{E}^{-1}\H(\cdot,\psi_{\lambda},\psi'),\psi,\psi^{\lambda}_{E}) \\
& + \H( \widehat{T}^{LC}(\psi_{\lambda},\psi'),\psi,\psi^{\lambda}_{E}) + \H(\widehat{T}^{LC}(\psi,\psi^{\lambda}_{E}),\psi_{\lambda},\psi') \},
\end{split}
\end{equation}
where $\hRic{}^{LC}$ is a Ricci tensor corresponding to the Courant algebroid connection $\hcDL$. One can now significantly simplify this expression. First, note that one has 
\begin{equation}
\frac{1}{2} \{ (\hcDL_{\psi_{\lambda}}\H)(\psi,\psi',\psi^{\lambda}_{E}) + (\hcDL_{\psi_{\lambda}}\H)(\psi',\psi,\psi^{\lambda}_{E})\} = (\hcDL_{\psi_{\lambda}} \H_{s})(\psi,\psi',\psi^{\lambda}_{E}),
\end{equation}
where $\H_{s}$ is the symmetrization of $\H$ in first two indices. Moreover, as $\hcDL$ is block diagonal and induced, one has to consider only the terms of $\H_{s}$ where the third input is of the form $(0,0,\zeta)$:
\begin{equation}
\begin{split}
\H_{s}((X,\Phi,\xi),(Y,\Phi',\eta),(0,0,\zeta)) = & \ \frac{1}{4} \{ H'_{0}(g_{0}^{-1}(\xi), Y, g_{0}^{-1}(\zeta)) + H'_{0}( g_{0}^{-1}(\eta),X,g_{0}^{-1}(\zeta)) \} \\
& + \frac{1}{2} \{ \< F'(g_{0}^{-1}(\zeta),X),\Phi'\>_{\g} + \<F'(g_{0}^{-1}(\zeta),Y),\Phi\>_{\g} \}.
\end{split}
\end{equation}
With the help of the formula (\ref{eq_covariantdifexpl}), one can now derive the expression:
\begin{equation} \label{eq_HsCDLexpl}
\begin{split}
(\hcDL_{\psi_{\lambda}} \H_{s})(\psi,\psi',\psi^{\lambda}_{E}) = & \ \frac{1}{4} \{ (\cDN_{e_{k}} H'_{0})(g_{0}^{-1}(\xi),Y,e^{k}_{0}) + (\cDN_{e_{k}}H'_{0})(g_{0}^{-1}(\eta),X,e_{0}^{k}) \} \\
& - \frac{1}{2} \{ \< (D'^{\dagger} F')(X), \Phi' \>_{\g} + \< ( D'^{\dagger} F')(Y), \Phi \>_{\g} \},
\end{split}
\end{equation}
for $\psi = (X,\Phi,\xi)$ and $\psi' = (Y,\Phi',\eta)$, where $\{e_{k}\}_{k=1}^{\dim{M}}$ is an arbitrary local frame on $M$, and $e^{k}_{0} = g_{0}^{-1}(e^{k})$. Moreover, the remaining five terms in (\ref{eq_Ric'expansioninH}) can be using (\ref{eq_H1}, \ref{eq_H3}) rewritten as:
\begin{equation}
\M(\psi,\psi') = \frac{1}{2} \H(\psi^{\lambda}_{E}, \psi^{\mu}_{E}, \psi') \{ \H(\psi_{\lambda},\psi_{\mu},\psi) - 2 \H(\psi_{\mu},\psi_{\lambda},\psi) \}. 
\end{equation}
It is easy to see that this is a symmetric tensor on $E'$. By plugging in (\ref{eq_Htensor}), and by taking $\psi = (X,\Phi,\xi)$, $\psi' = (Y,\Phi',\eta)$, one finds the explicit formula:
\begin{equation} \label{eq_Mexpl}
\begin{split}
\M(\psi,\psi') = & \ -\frac{1}{3} \<\io_{X}H'_{0},\io_{Y}H'_{0}\>_{g_{0}} + \frac{1}{6} \< \io_{g_{0}^{-1}(\xi)} H'_{0}, \io_{g_{0}^{-1}(\eta)}H'_{0}\>_{g_{0}} - \frac{1}{6} \<\Phi,\Phi'\>_{\g} \\
& - \frac{1}{4} \< \<F', \io_{X +g_{0}^{-1}(\xi)} H'_{0}\>_{g_{0}}, \Phi'\>_{\g}  - \frac{1}{4} \< \< F', \io_{Y + g_{0}^{-1}(\eta)} H'_{0}\>_{g_{0}}, \Phi \>_{\g} \\
& + \frac{1}{2} \< F'(X,e_{m}), F'(Y,e^{m}_{0}) \>_{\g} + \frac{1}{4} \< F'(e_{k},e_{m}), \Phi\>_{\g} \< F'(e^{k}_{0},e^{m}_{0}), \Phi' \>_{\g} \\
&  - \frac{1}{8} \<F'(X,e_{m}), F'(g_{0}^{-1}(\eta),e^{m}_{0}) \>_{\g} - \frac{1}{8} \<F'(Y,e_{m}), F'(g_{0}^{-1}(\xi), e^{m}_{0}) \>_{\g}. 
\end{split}
\end{equation}
Finally, it is straightforward to see that $\hRic{}^{LC}$ can be directly related to the ordinary Ricci tensor $\Ric_{0}^{LC}$ corresponding to the metric $g_{0}$ as
\begin{equation}
\hRic{}^{LC}(\psi,\psi') = \Ric_{0}^{LC}( \rho'(\psi), \rho'(\psi')),
\end{equation}
for all $\psi,\psi' \in \Gamma(E')$. Altogether, the formula (\ref{eq_Ric'expansioninH}) can be now written in the form
\begin{equation} \label{eq_Ric'ascDHsandM}
\Ric'(\psi,\psi') = \Ric^{LC}_{0}(\rho'(\psi),\rho'(\psi')) + (\hcDL_{\psi_{\lambda}} \H_{s})(\psi,\psi',\psi^{\lambda}_{E}) + \M(\psi,\psi') 
\end{equation}
Before the calculation of the scalar curvatures, note that by construction, the scalar curvatures $(\RS'_{E'},\RS'_{\gm'})$ corresponding to the original Levi-Civita connection $\cD'$ are the same as a pair $(\widehat{\RS}'_{E'}, \widehat{\RS}'_{\G'})$ corresponding to the twisted connection $\hcD$ we have worked with. This follows from (\ref{eq_hcD'twist}) and the covariance of all involved objects under Courant algebroid isomorphisms. Thus
\begin{equation}
\RS'_{\gm'} = \Ric'( \G'^{-1}(\psi^{\lambda}), \psi_{\lambda}), \; \; \RS'_{E'} = \Ric'( g_{E'}^{-1}(\psi^{\lambda}), \psi_{\lambda}).
\end{equation}
Plugging (\ref{eq_Ric'ascDHsandM}) into these and using (\ref{eq_HsCDLexpl}, \ref{eq_Mexpl}) gives the expressions:
\begin{align}
\label{eq_RS'1} \RS'_{\gm'} & = \RS(g_{0}) - \frac{1}{2} \<H'_{0},H'_{0}\>_{g_{0}} + \frac{1}{2} \dal F' , F' \dar + \frac{1}{6} \dim{\g}, \\
\label{eq_RS'2} \RS'_{E'} & = - \frac{1}{6} \dim{\g}. 
\end{align}
Finally, we have to study the Ricci compatibility of $\cD'$ with $\gm'$. It is easy to see that this is equivalent to the Ricci compatibility of $\hcD'$ with $\G'$. Let $E' = V'_{+} \oplus V'_{-}$ be the decomposition of $E'$ induced by the generalized metric $\G'$. Explicitly, one has 
\begin{align}
\label{eq_V'+iso} \Gamma(V'_{+}) & = \{ (X,0,g_{0}(X)) \; | \; X \in \vf{} \} \cong \vf{}, \\
\label{eq_V'-iso} \Gamma(V'_{-}) & = \{ (X,\Phi,-g_{0}(X)) \; | \; (X,\Phi) \in \vf{} \oplus \Gamma(\g_{P}) \} \cong \vf{} \oplus \Gamma(\g_{P}). 
\end{align}
The Ricci compatibility of $\Ric'$ with $\G'$ is thus equivalent to the vanishing of $\Ric'_{+-}$ defined by 
\begin{equation}
\Ric'_{+-}(X,(Y,\Phi)) = \Ric'\big( (X,0,g_{0}(X)), (Y,\Phi,-g_{0}(Y))\big). 
\end{equation}
Plugging in (\ref{eq_Ric'ascDHsandM}) and (\ref{eq_HsCDLexpl}, \ref{eq_Mexpl}), one finds the following expressions:
\begin{align}
\label{eq_Ric'pm1} \Ric'_{+-}(X,(Y,0)) = & \ - \frac{1}{2} (\delta_{g_{0}}H'_{0})(X,Y) - \frac{1}{2} \< \io_{X}H'_{0}, \io_{Y}H'_{0} \>_{g_{0}} + \frac{1}{2} \dal \io_{X}F', \io_{Y}F' \dar, \\
\label{eq_Ric'pm2} \Ric'_{+-}(X,(0,\Phi)) = & \ - \frac{1}{2} \< (D'^{\dagger}F')(X),\Phi \>_{\g} - \frac{1}{2} \< \<F', \io_{X}H'_{0} \>_{g_{0}}, \Phi \>_{\g}. 
\end{align}
\subsection{Introducing the dilaton}
Now, we have to think how to encode the dilaton $\phi_{0}$ into the connection. To do so, first note that having a Levi-Civita connection $\hcD'$, one can define a new Levi-Civita connection $\hcD^{\K}$ using the formula
\begin{equation} \label{eq_hcDK}
\hcD^{\K}_{\psi}\psi' =\hcD'_{\psi}\psi' + g_{E'}^{-1} \K(\psi,\psi',\cdot),
\end{equation}
where $\K \in \T_{3}^{0}(E')$ has to satisfy the conditions similar to (\ref{eq_H1} - \ref{eq_H3}) for $\H$. 
\begin{align}
\label{eq_K1} \K(\psi,\psi',\psi'') + \K(\psi,\psi'',\psi') & = 0, \\
\label{eq_K2} \K(\psi,\psi',\tau'_{0}(\psi'')) + \K(\psi,\psi'',\tau'_{0}(\psi')) &= 0, \\
\label{eq_K3} \K(\psi,\psi',\psi'') + cyclic(\psi,\psi',\psi'') & = 0. 
\end{align}
One can now express the Ricci tensor $\Ric'_{\K}$ for $\hcD^{K}$ using a formula analogous to (\ref{eq_Ric'expansioninH}). One finds 
\begin{equation} \label{eq_RicKasRic'}
\begin{split}
\Ric'_{\K}(\psi,\psi') = \Ric'(\psi,\psi') + & (\hcD'_{\psi_{\lambda}} \K_{s})(\psi,\psi',\psi^{\lambda}_{E}) + \frac{1}{2} \{ (\hcD'_{\psi} \K')(\psi') + (\hcD'_{\psi'} \K')(\psi) \} \\
& - \K'(\psi_{\lambda}) \K_{s}(\psi,\psi',\psi^{\lambda}_{E}) + \M_{\K}(\psi,\psi'), 
\end{split}
\end{equation}
where $\K_{s}$ is the symmetrization of $\K$ in first two inputs, $\K' \in \Omega^{1}(E')$ is the partial trace $\K'(\psi) = \K( \psi^{\lambda}_{E}, \psi_{\lambda},\psi)$, and $\M_{\K}$ is a symmetric tensor on $E'$ defined by 
\begin{equation}
\M_{\K}(\psi,\psi') = \frac{1}{2} \K(\psi^{\lambda}_{E},\psi^{\mu}_{E},\psi') \{ \K(\psi_{\lambda},\psi_{\mu},\psi) - 2 \K( \psi_{\mu}, \psi_{\lambda}, \psi) \}.
\end{equation}
To find the expression for the scalar curvatures $\RS_{\G'}^{\K}$ and $\RS_{E'}^{\K}$, it is convenient to use the splitting $E' = V'_{+} \oplus V'_{-}$ induced by the generalized metric $\G'$. By definition, $\hcD'$ preserves the subbundles $V'_{\pm}$, and thus defines a pair of connections $\cD^{\pm}: \Gamma(V'_{\pm}) \times \Gamma(V'_{\pm}) \rightarrow \Gamma(V'_{\pm})$. Moreover, (\ref{eq_K1} - \ref{eq_K2}) imply that $\K$ is non-trivial only if all its inputs come from the same subbundle. Let $\K_{\pm}$ denote its restriction to $V'_{\pm}$. Clearly $(\K_{s})_{\pm} = (\K_{\pm})_{s}$. Finally, the generalized metric $\G'$ restricts to a pair $\G'_{\pm}$ of fiber-wise metrics on $V'_{\pm}$. The partial trace $\K'$ decomposes as 
\begin{equation}
\K'(\psi) = \K'_{+}(\psi_{+}) - \K'_{-}(\psi_{-}),  
\end{equation}
where $\psi = \psi_{+} + \psi_{-}$ is a decomposition with respect to $E' = V'_{+} \oplus V'_{-}$, and $\K'_{\pm}$ are the partial traces of $\K_{\pm}$ taken over $V'_{\pm}$ using the fiber-wise metric $\G'_{\pm}$. One can now prove the formulae
\begin{align}
\label{eq_RSGKasjinak} \RS_{\G'}^{\K} &= \RS'_{\G'} + 2 \Div_{\cD^{+}} (\K'_{+}) - 2 \Div_{\cD^{-}} (\K'_{-}) - \rVert \K'_{+} \rVert^{2}_{\G'_{+}} - \rVert \K'_{-} \rVert^{2}_{\G'_{-}}, \\
\label{eq_RSEKasjinak} \RS_{E'}^{\K} &= \RS'_{E'} + 2 \Div_{\cD^{+}} (\K'_{+}) + 2 \Div_{\cD^{-}}( \K'_{-}) - \rVert \K'_{-} \rVert^{2}_{\G'_{-}} + \rVert \K'_{+} \rVert^{2}_{\G'_{-}}, 
\end{align}
where $\Div_{\cD^{\pm}}$ is a covariant divergence on $V'_{\pm}$ defined using the fiber-wise metrics $\G'_{\pm}$. This expressions are useful, having the isomorphisms (\ref{eq_V'+iso}, \ref{eq_V'-iso}). Instead of working on $V'_{\pm}$, we may use the induced objects on $TM$ and $TM \oplus \g_{P}$, respectively. We will denote the induced objects by the same symbols. First, the induced connection $\cD^{+}$ takes the form
\begin{equation}
\cD^{+}_{X}Y = \cDN_{X}Y - \frac{1}{6} g_{0}^{-1}H'_{0}(X,Y,\cdot),
\end{equation}
for all $X,Y \in \vf{}$. The connection $\cD^{-}$ is more complicated and it can be written as
\begin{equation}
\cD^{-}_{(X,\Phi)} = \bm{\cDN_{X} + \frac{1}{6} g_{0}^{-1}H'_{0}(X,\star,\cdot) + \frac{1}{2} g_{0}^{-1} \<\Phi,F'(\star)\>_{\g}}{ \frac{1}{2} g_{0}^{-1} \<F'(X),\star \>_{\g}}{ - F'(X,\star)}{D'_{X} - \frac{1}{3} [\Phi,\star]_{\g}}. 
\end{equation}
The induced fiber-wise metric $\G'_{+}$ on $TM$ is $\G'_{+} = 2 g_{0}$, whereas $\G'_{-}$ takes the form 
\begin{equation}
\G'_{-}((X,\Phi),(Y,\Phi)) = 2 g_{0}(X,Y) - \< \Phi,\Phi'\>_{\g}. 
\end{equation} 
To determine the tensor $\K$, it suffices to define the pair of induced tensors $\K_{+} \in \T_{3}^{0}(M)$ and $\K_{-} \in \T_{3}^{0}(TM \oplus \g_{P})$. Let $\phi_{0} \in \cif$ be any scalar function. Define 
\begin{align}
\label{eq_Kdef1} \K_{+}(X,Y,Z) = (4 / (\dim{M} - 1)) \{ g_{0}(X,Y) \< d\phi_{0},Z\> - g_{0}(X,Z) \< d\phi_{0},Y\> \}, \\
\label{eq_Kdef2} \K_{-}((X,\Phi),(Y,\Phi'),(Z,\Phi'')) = -\K_{+}(X,Y,Z),
\end{align}
for all $X,Y,Z \in \vf{}$ and $\Phi,\Phi',\Phi'' \in \Gamma(\g_{P})$. It is now a straightforward calculation to prove that by plugging into (\ref{eq_RSGKasjinak}, \ref{eq_RSEKasjinak}), one obtains the final expression
\begin{equation} \label{eq_RSKasjinak}
\RS^{\K}_{\G'} = \RS'_{\G'} + 4 \Delta_{g_{0}}(\phi_{0}) - 4 \< d\phi_{0}, d\phi_{0} \>_{g_{0}}, \; \; \RS^{\K}_{E'} = \RS'_{E'}. 
\end{equation}
We see that the addition of $\K$ adds the correct kinetic terms for dilaton $\phi_{0}$ into the first scalar curvature, without changing the Courant-Ricci scalar. We only have to check how this choice of $\K$ modifies the Ricci compatibility conditions. Note that $\K'(Z,\Phi'',\zeta) = 2 \< d\phi_{0}, Z\>$.  Hence 
\begin{equation}
\begin{split}
\frac{1}{2} \{ (\hcD'_{\psi}\K')(\psi') + (\hcD'_{\psi'}\K')(\psi) \} = & \  ( \cDN_{X} d\phi_{0})(Y) + (\cDN_{Y} d\phi_{0})(X) \\
& - H'_{0}(X,Y,g_{0}^{-1}(d\phi_{0})) - \< F'(g_{0}^{-1}(d\phi_{0}),X), \Phi \>_{\g},
\end{split}
\end{equation}
for $\psi = (X,0,g_{0}(X))$ and $\psi' = (Y,\Phi,-g_{0}(Y))$. All other terms in (\ref{eq_RicKasRic'}) are easily seen to vanish for this combination of $(\psi,\psi')$. We thus conclude that 
\begin{align}
\label{eq_RicKpm1} \Ric^{\K}_{+-}(X,(Y,0)) = & \ \Ric'_{+-}(X,(Y,0)) + (\cDN_{X} d\phi_{0})(Y) + (\cDN_{Y} d\phi_{0})(X) \\
& - H'_{0}(X,Y,g_{0}^{-1}(d\phi_{0})), \nonumber \\
\label{eq_RicKpm2} \Ric^{\K}_{+-}(X,(0,\Phi)) =& \ \Ric'_{+-}(X,(0,\Phi)) - \< F'(g_{0}^{-1}(d\phi_{0}),X), \Phi\>_{\g}. 
\end{align}
This indeed adds the correct dilaton contribution to the Ricci compatibility condition $\Ric_{+-}^{\K} = 0$. 
\begin{rem}
The choice of $\K$ above does not come out of the blue. In fact, it corresponds to the deformation of Levi-Civita connections via so called \textit{Weyl endomorphisms}, see \cite{Garcia-Fernandez:2016ofz}. In more detail, let $\tcD'$ and $\hcD'$ be two Levi-Civita connections, such that their covariant divergencies are related using a given element $\xi \in \Gamma(E'^{\ast})$ as
\begin{equation} \label{eq_divKrel}
\Div_{\tcD'}(\psi) = \Div_{\hcD'}(\psi) - \xi(\psi),
\end{equation}
for all $\psi \in \Gamma(E')$. Then there is $\K \in \Omega^{1}(E') \otimes \Omega^{2}(E')$ satisfying (\ref{eq_K1} - \ref{eq_K3}), such that $\tcD' = \hcD^{\K}$ as in (\ref{eq_hcDK}). Importantly, it can be shown that neither of the quantities $\RS^{\K}_{\G'}$, $\RS^{\K}_{E'}$ and $\Ric^{\K}_{+-}$ depends on the choice of such $\K$. This can be seen directly from (\ref{eq_RicKasRic'}) and it is (stated slightly differently) proved in \cite{Garcia-Fernandez:2016ofz}. In our case, we define $\K$ so that the $\xi \in \Gamma(E'^{\ast})$ defined by (\ref{eq_divKrel}) takes the form
\begin{equation}
\xi(\psi) = 2 \cdot \< d\phi_{0}, \rho'(\psi) \>,
\end{equation}
for all $\psi \in \Gamma(E')$. See \cite{Garcia-Fernandez:2016ofz} for so called \emph{Weyl gauge fixing}. 
\end{rem}
\subsection{Main statement}
Combining the previous two sections, we now arrive to the main result of this section. Its proof is a summarization of the above calculations. 
\begin{theorem}\label{thm_eomheterotic}
Let $\cD'$ be the Courant algebroid connection on $(E',\rho',\<\cdot,\cdot\>_{E'},[\cdot,\cdot]_{E'})$ defined as 
\begin{equation}
\cD'_{\psi} \psi' = e^{\C}( \hcD^{\K}_{e^{-\C}(\psi)} e^{-\C}(\psi') ), 
\end{equation}
for all $\psi,\psi' \in \Gamma(E')$, where $\C \in \End(E')$ was defined in (\ref{eq_defG'andC}), and $\hcD^{\K}$ has the form 
\begin{equation}
\hcD^{\K}_{\psi}\psi' = \hcD'_{\psi}\psi' + g_{E'}^{-1} \K(\psi,\psi',\cdot),
\end{equation}
for all $\psi,\psi' \in \Gamma(E')$, and $\hcD'$ is defined by (\ref{eq_hcD'1} - \ref{eq_hcD'3}). The tensor $\K$ is defined by (\ref{eq_Kdef1},\ref{eq_Kdef2}). Note that such $\cD'$ is a Levi-Civita connection on $E'$ with respect to the generalized metric $\gm'$. 

Then $(g_{0},B_{0},\phi_{0},\vartheta)$ satisfy the equations of motion (\ref{eq_thmEOM2}), if and only if $\RS'_{\gm'} = 2 \Lambda_{0} + \frac{1}{6} \dim{\g}$ and $\cD'$ is Ricci compatible with $\gm'$, that is $\Ric'(V'_{+},V'_{-}) = 0$. Moreover, one has $\RS'_{E'} = -\frac{1}{6} \dim{\g}$. 
\end{theorem}
\begin{proof}
For scalar curvatures, we have $\RS'_{\gm'} = \RS^{\K}_{\G'}$ and $\RS'_{E'} = \RS^{K}_{E'}$, where the functions on the right-hand side are given by (\ref{eq_RSKasjinak}) and (\ref{eq_RS'1}, \ref{eq_RS'2}). We thus obtain 
\begin{equation}
\RS'_{\gm'} = \beta'_{\phi_{0}} + ( 2 \Lambda_{0} + \frac{1}{6} \dim{\g}), \; \; \RS'_{E'} = - \frac{1}{6} \dim{\g}. 
\end{equation}
Next, $\cD'$ is Ricci compatible with $\gm'$, if and only if $\Ric^{\K}_{+-} = 0$. Plugging (\ref{eq_Ric'pm1}, \ref{eq_Ric'pm2}) into (\ref{eq_RicKpm1}, \ref{eq_RicKpm2}), one finds
\begin{align}
\Ric^{\K}_{+-}(X,(Y,0)) & = \beta'(g_{0})(X,Y) - \beta'(B_{0})(X,Y), \\
\Ric^{\K}_{+-}(X,(0,\Phi)) & = - \< \beta'(\vartheta), \Phi \>_{\g},
\end{align}
for all $X,Y \in \vf{}$ and $\Phi \in \Gamma(\g_{P})$. The statement of the theorem now follows easily. 
\end{proof}
\section{Reduction procedure} \label{sec_reduction}
In the previous section, we have established a geometrical description of the equations of motion for both the actions (\ref{eq_S}) and (\ref{eq_S0}). Courant algebroids over a principal $G$-bundle $\pi: P \rightarrow M$ can be under some conditions reduced to Courant algebroids over the base manifold $M$. This was in great detail described in \cite{Bursztyn2007726} and in present context also in \cite{Baraglia:2013wua} or \cite{2015LMaPh.tmp...53S}. Moreover, one can discover necessary and sufficient conditions for a reduction of generalized metrics and Levi-Civita connections. We will now briefly recall all necessary notions. 
\subsection{Equivariant Courant algebroids}
Let $(E,\rho,\<\cdot,\cdot\>_{E},[\cdot,\cdot]_{E})$ be a Courant algebroid over a principal $G$-bundle $\pi: P \rightarrow M$. Let $\Re: \g \rightarrow \Gamma(E)$ be a $\R$-linear map preserving the bracket:
\begin{equation}
\Re( [x,y]_{\g}) = [\Re(x),\Re(y)]_{E}, 
\end{equation}
for all $x,y \in \g$, and covering the infinitesimal action $\#: \g \rightarrow \vfP{}$, that is the diagram
\begin{equation}
\begin{tikzcd}
\g \arrow[r,"\Re"] \arrow[rd,"\#"'] & \Gamma(E) \arrow[d,"\rho"] \\
& \vfP{}
\end{tikzcd}
\end{equation}
commutes. It follows that $x \tr \psi \equiv [\Re(x),\psi]_{E}$ defines an infinitesimal action of $\g$ on $\Gamma(E)$. One can show that it also defines an infinitesimal Lie algebra action on a total space manifold $E$. It always integrates \emph{locally} to an action of Lie group $G$ acting via vector bundle isomorphisms over the principal bundle action $R: P \times G \rightarrow P$. This leads one to the following definition:

\begin{definice}
Let $\Re: \g \rightarrow \Gamma(E)$ be a map defined as above. $(E,\rho,\<\cdot,\cdot\>_{E},[\cdot,\cdot]_{E})$ is called an \textbf{equivariant Courant algebroid} if $\tr$ integrates globally to a Lie group action $\fR: E \times G \rightarrow E$ covering the principal bundle action $R$, that is 
\begin{enumerate}
\item For each $g \in G$, $(\fR_{g},R_{g})$ is a vector bundle isomorphism. 
\item For each $x \in \g$ and $\psi \in \Gamma(E)$, one has $x \tr \psi = \ddt \hspace{5mm} \fR_{e^{-tx}}(\psi)$, 
\end{enumerate}
where $\{ \fR_{g}(\psi)\}(m) := \fR_{g}\{ \psi( R_{g^{-1}}(m)) \}$, for all $g \in G$ and $m \in M$. 
\end{definice}
There are several remarks in order. First, observe that one has 
\begin{equation}
\D \<\Re(x),\Re(y)\>_{E} = [\Re(x),\Re(y)]_{E} + [\Re(y),\Re(x)]_{E} = \Re( [x,y]_{\g} + [y,x]_{\g}) = 0. 
\end{equation}
Recall that $\D = \rho^{\ast} \circ d$. Whenever $\rho$ is fiber-wise surjective, the above condition implies that $\<\Re(x),\Re(y)\>_{E}$ is locally constant and thus defines a symmetric bilinear form $(\cdot,\cdot)_{\g}$ on $\g$ via 
\begin{equation} \label{eq_(.,.)}
(x,y)_{\g} = -\<\Re(x),\Re(y) \>_{\g},
\end{equation}
for all $x,y \in \g$, where we for simplicity assume that $M$ is connected. Moreover, it follows that $(\cdot,\cdot)_{\g}$ is ad-invariant with respect to $[\cdot,\cdot]_{\g}$. It is now straightforward to see that $(\fR_{g},R_{g})$ forms a Courant algebroid isomorphism of $(E,\rho,\<\cdot,\cdot\>_{E},[\cdot,\cdot]_{E})$ for each $g \in G$.  We can now make use of the simple observation making things very easy in our case. For the original idea see \cite{2015LMaPh.tmp...53S}. We state it here without the proof.
\begin{lemma} \label{lem_equispecial}
Let $E = \gTP$ be equipped with the standard Courant algebroid structure given by the $H$-twisted Dorfman bracket as in Example \ref{ex_dorfman}. Assume that $\Re: \g \rightarrow \Gamma(\gTP)$ makes $E$ into an equivariant Courant algebroid. Let $A \in \Omega^{1}(P,\g)$ be any principal bundle connection. 

Then there exists a $2$-form $B \in \Omega^{2}(P)$, such that $\Re' = e^{-B} \circ \Re$ makes $(H+dB)$-twisted Dorfman bracket on $\gTP$ into an equivariant Courant algebroid, and $\Re'$ has the form
\begin{equation}
\Re'(x) = (\#{x}, - \frac{1}{2} (A,x)_{\g}). 
\end{equation}
Moreover, one has $(\cdot,\cdot)'_{\g} = (\cdot,\cdot)_{\g}$, and the $3$-form $H' \equiv H + dB$ has to take the form 
\begin{equation}
H' = \pi^{\ast}(H_{0}) + \frac{1}{2} CS_{3}(A),
\end{equation}
where $H_{0} \in \Omega^{3}(M)$ satisfies the equation $dH_{0} + \frac{1}{2} (F \^ F)_{\g} = 0$, and $CS_{3}(A)$ is defined as 
\begin{equation}
CS_{3}(A) = ( dA \^ A )_{\g} + \frac{1}{3} ( [A \^ A]_{\g} \^ A )_{\g}. 
\end{equation}
\end{lemma}
This lemma has important consequences. First, note that the action $\tr'$ induced by $\Re'$ is
\begin{equation}
\begin{split}
x \tr' (Y,\eta) = & \ [(\#{x}, -\frac{1}{2}(A,x)_{\g}), (Y,\eta)]'_{E} = \big( [\#{x},Y], \Li{\#{x}}\eta + \frac{1}{2}\io_{Y}d(A,x)_{\g} - H'( \#{x},Y,\cdot) \big) \\
= & \ \big( [\#{x},Y], \Li{\#{x}}\eta + \frac{1}{2} \{ \io_{Y}d(A,x)_{\g} - \io_{Y}\io_{\#{x}} CS_{3}(A) \} \big) \\
= & \ \big( [\#{x},Y], \Li{\#{x}}\eta \big),
\end{split}
\end{equation}
where we have used the fact that $\io_{\#{x}}CS_{3}(A) = (dA,x)_{\g}$. In other words, the action $\fR'$ integrating $\tr'$ is an ordinary right translation on the $\gTP$ induced by the principal bundle action $R$ on $P$. Next, we see that $\gTP$ with the $H$-twisted Dorfman bracket is equivariant with given $(\cdot,\cdot)_{\g}$, if and only if the \textbf{Pontriyagin class} $[(F \^ F)_{\g}]_{dR}$ vanishes. Note that $(\cdot,\cdot)_{\g}$ can be zero. 

Finally, we may consider choosing a different principal bundle connection $A' \in \Omega^{1}(P,\g)$. There thus exists a unique $\varrho \in \Omega^{1}(M,\g_{P})$, such that $X'^{h} = X^{h} - j(\varrho(X))$, see (\ref{eq_horliftchange}). Let 
\begin{equation}
\Re(x) = (\#{x}, - \frac{1}{2}(A,x)_{\g}), \; \; \Re'(x) = (\#{x}, - \frac{1}{2}(A',x)_{\g}). 
\end{equation}
By previous lemma, there exists a (non-unique) $2$-form $B \in \Omega^{2}(P)$, such that $\Re'(x) = e^{-B}( \Re(x))$. We thus have to find a solution $B$ to the equation $\frac{1}{2} (A'-A,x)_{\g} = -\io_{\#{x}}B$. It follows that $B$ must be a $G$-invariant $2$-form, that is $\Li{\#{x}}B = 0$. In fact, it is easy to find the general solution:
\begin{equation} \label{eq_BasB0plusproduct}
B = \pi^{\ast}(B_{0}) + \frac{1}{2} (A' \^ A)_{\g}
\end{equation}
As $B$ is $G$-invariant, it can be written with respect to the decomposition $\mathfrak{X}_{G}(P) \cong \vf{} \oplus \Gamma(\g_{P})$ induced by the connection $A$. One can write $B$ in the block form as 
\begin{equation}
B = \bm{B_{0}}{\frac{1}{2} \varrho^{T} s}{ -\frac{1}{2} s \varrho}{0},
\end{equation}
where $s \in \Hom(\g_{P},\g^{\ast}_{P})$ is the map induced by a symmetric ad-invariant bilinear form $(\cdot,\cdot)_{\g}$. Note that for $(\cdot,\cdot)_{\g} = \<\cdot,\cdot\>_{\g}$ and $\varrho = \vartheta$, we obtain exactly the same form of $B$ as in (\ref{eq_gBrelevant}). 
\subsection{Reduction of Courant algebroids}
Let us assume that $E = \gTP$ is equipped with the $H$-twisted Dorfman bracket and $A \in \Omega^{1}(P,\g)$ is a principal bundle connection. Let $\Re$ be of the special form given by Lemma \ref{lem_equispecial}, that is
\begin{equation} \label{eq_ReHrelevant}
\Re(x) = (\#{x}, - \frac{1}{2}(A,x)_{\g}), \; \; H = \pi^{\ast}(H_{0}) + \frac{1}{2} CS_{3}(A),
\end{equation}
for all $x \in \g$, where $(\cdot,\cdot)_{\g}$ is any invariant symmetric bilinear form on $\g$. The reduced Courant algebroid $E_{red}$ over $M$ is constructed as follows. Note that $\Re$ can be viewed as a vector bundle map from a trivial vector bundle $P \times \g$ to $E$ which coincides with the original map on constant sections. It is fiber-wise injective and we can thus define a subbundle $K \subseteq \gTP$ as $K = \Re(P \times \g)$. It is easy to see that $K$ is invariant with respect to the action $\fR$ integrating $\Re$. Indeed, let $\Phi \in \Gamma(P \times \g) \cong C^{\infty}(P,\g)$. Then one has 
\begin{equation}
\begin{split}
x \tr \Re(\Phi) = & \ [\Re(x), \Phi^{\alpha} \Re(t_{\alpha}) ]_{E} = \Phi^{\alpha} [ \Re(x), \Re(t_{\alpha})]_{E} + (\rho(\Re(x)).\Phi^{\alpha}) \Re(t_{\alpha}) \\
= & \ \Phi^{\alpha} \Re( [x,t_{\alpha}]_{\g}) + (\#{x}.\Phi^{\alpha}) \Re(t_{\alpha}) \\
= & \ \Re( [x,\Phi]_{\g} + \#{x}.\Phi ).
\end{split}
\end{equation}
This proves that $x \tr \Re(\Phi) \in \Gamma(K)$. Moreover, we see that $\Re(\Phi) \in \Gamma_{G}(K)$, if and only if $\Phi \in C^{\infty}_{Ad}(P,\g) \cong \Gamma(\g_{P})$. In other words, $\Re$ induces a $\cif$-module isomorphism $\Gamma(\g_{P}) \cong \Gamma_{G}(K)$. Explicitly
\begin{equation}
\Gamma_{G}(K) = \{ (j(\Phi), -\frac{1}{2}(A,\Phi)_{\g}) \; | \; \Phi \in \Gamma(\g_{P}) \},
\end{equation}
where $j: \Gamma(\g_{P}) \rightarrow \mathfrak{X}_{G}(P)$ is the inclusion of $\g_{P}$ into $G$-invariant vector fields on $P$. Next, one considers the orthogonal complement $K^{\perp}$. This is again a $G$-invariant subbundle, as one has 
\begin{equation}
\< x \tr \psi, \Re(y) \>_{E} = \< [\Re(x),\psi]_{E}, \Re(y) \>_{E} = \#{x}.\<\psi,\Re(y)\>_{E} - \< \psi, \Re([x,y]_{\g})\>_{E} = 0,
\end{equation}
for all $x,y \in \g$, whenever $\psi \in \Gamma(K^{\perp})$. Using the explicit form of $\Re$, one can see that 
\begin{equation}
\Gamma_{G}(K^{\perp}) = \{ (X^{h} + j(\Phi), \pi^{\ast}(\xi) + \frac{1}{2} (A, \Phi)_{\g}) \; | \; (X,\Phi,\xi) \in \Gamma(TM \oplus \g_{P} \oplus T^{\ast}M) \}. 
\end{equation}
Finally, for general $(\cdot,\cdot)_{\g}$, the intersection $K \cap K^{\perp}$ is not trivial. Instead, one finds 
\begin{equation}
\Gamma_{G}(K \cap K^{\perp}) = \{ (j(\Phi), 0) \; | \; \Phi \in \Gamma(\mathfrak{n}_{P}) \},
\end{equation}
where $\n_{P} \subseteq \g_{P}$ is a subbundle of $\g_{P}$, such that $\Gamma(\n_{P}) \cong C^{\infty}_{Ad}(P,\n)$, and $\n \equiv \ker(s)$. The reduced Courant algebroid $E_{red}$ is defined on the quotient vector bundle:
\begin{equation}
E_{red} = \frac{K^{\perp} / G}{ (K^{\perp} \cap K) / G}.
\end{equation}
On the level of sections, we have $\Gamma(E_{red}) = \Gamma_{G}(K^{\perp}) / \Gamma_{G}(K \cap K^{\perp})$. The Courant algebroid structure on $E_{red}$ is naturally induced by the original structure on $E$. 

Let $\fPsi': \Gamma(TM \oplus \g_{P} \oplus T^{\ast}M) \rightarrow \Gamma_{G}(K^{\perp})$ be the map defined as 
\begin{equation}
\fPsi'( X, \Phi,\xi) = (X^{h} + j(\Phi), \pi^{\ast}(\xi) + \frac{1}{2} (A,\Phi)_{\g}), 
\end{equation}
for all $(X,\Phi,\xi) \in \Gamma(TM \oplus \g_{P} \oplus T^{\ast}M)$. It follows from the above discussion that the map $\fPsi \in \Hom(E',E_{red})$ defined as $\fPsi(X,[\Phi],\xi) = [\fPsi'(X,\Phi,\xi)]$ forms a vector bundle isomorphism, where $E'$ is defined as $E' = TM \oplus (\g_{P}/ \n_{P}) \oplus T^{\ast}M$. Square brackets denote equivalence classes in $\Gamma(\g_{P}/\n_{P})$ and $\Gamma(E_{red})$, respectively. Let us now examine the resulting induced Courant algebroid structure on $E'$. For the anchor $\rho'$, we get
\begin{equation}
\rho'(X,[\Phi],\xi) \circ \pi = T(\pi) \circ \rho( \fPsi'(X,\Phi,\xi)) = T(\pi)(X^{h}) = X \circ \pi,
\end{equation}
and thus $\rho' \in \Hom(E',TM)$ is just a projection onto $TM$. For the pairing $\<\cdot,\cdot\>_{E'}$, one obtains
\begin{equation}
\begin{split}
\< (X,[\Phi],\xi), (Y,[\Phi'],\eta) \>_{E'} \circ \pi = & \ \< \fPsi'(X,\Phi,\xi), \fPsi'(Y,\Phi',\eta) \>_{E} \\
= & \ \{ \< \eta,X\> + \<\xi,Y\> + (\Phi,\Phi')_{\g} \} \circ \pi
\end{split}
\end{equation}
We find that $\<(X,[\Phi],\xi), (Y,[\Phi'],\eta) \>_{E'} = \<\eta,X\> + \<\xi,Y\> + (\Phi,\Phi')_{\g}$. This is a well defined non-degenerate $\cif$-linear symmetric form, hence a fiber-wise metric on $E'$. Next, one gets
\begin{equation} \label{eq_bracketE'gen}
\begin{split}
[(X,[\Phi],\xi),(Y,[\Phi'],\eta)]_{E'} = &\big( [X,Y], [D_{X}\Phi' - D_{Y}\Phi - [\Phi,\Phi']_{\g} -F(X,Y) ], \Li{X}\eta - \io_{Y}d\xi \\
& -H_{0}(X,Y,\cdot) - (F(X),\Phi')_{\g} + (F(Y),\Phi)_{\g} + (D\Phi,\Phi')_{\g}
\big).
\end{split}
\end{equation}
To see that this bracket is well defined, notice that the covariant derivative $D_{X}$ preserves the subbundle $\n_{P}$. This follows from the metric compatibility of $D$ and $(\cdot,\cdot)_{\g}$: 
\begin{equation}
X.(\Phi,\Phi')_{\g} = ( D_{X}\Phi, \Phi')_{\g} + (\Phi, D_{X}\Phi')_{\g},
\end{equation}
for all $X \in \vf{}$ and $\Phi,\Phi' \in \Gamma(\g_{P})$. This equation can be derived from the axiom (\ref{eq_killingequation}) for $E$. Moreover, the ad-invariance of $(\cdot,\cdot)_{\g}$ implies that $\Gamma(\n_{P}) \subseteq \Gamma(\g_{P})$ is an ideal with respect to $[\cdot,\cdot]_{\g}$. These two properties ensure that the right-hand side of (\ref{eq_bracketE'gen}) does not depend on the choice of the representatives $\Phi$ and $\Phi'$. It is a straightforward verification that $(E',\rho',\<\cdot,\cdot\>_{E'},[\cdot,\cdot]_{E'})$ forms a Courant algebroid. 

To conclude this subsection, note that we can choose $(\cdot,\cdot)_{\g} = \<\cdot,\cdot\>_{\g}$. In this case $\n_{P} = 0$, and thus $E' = TM \oplus \g_{P} \oplus T^{\ast}M$. The bracket $[\cdot,\cdot]_{E'}$ is precisely the one (\ref{eq_E'bracket}) we have used in the previous section. Naturally, this is exactly how we have originally obtained such a Courant algebroid. Moreover, in this case $K \cap K^{\perp} = 0$, which implies that we have a decomposition 
\begin{equation}
\Gamma_{G}(E) = \Gamma_{G}(K^{\perp}) \oplus \Gamma_{G}(K) \cong \Gamma(E') \oplus \Gamma(\g_{P}). \end{equation}
As $\Gamma_{G}(E)$ is involutive with respect to $[\cdot,\cdot]_{E}$, it is convenient to write this bracket with respect to the above decomposition. For $(\psi,\Phi) \in \Gamma(E' \oplus \g_{P})$ and $(\psi',\Phi') \in \Gamma(E' \oplus \g_{P})$, one finds
\begin{equation}
[(\psi,\Phi), (\psi',\Phi')]_{E} = ( [\psi,\psi']_{E'} - [ \mathbf{i}(\Phi), \mathbf{i}(\Phi')]_{E'}, \mathbf{p}\{ [\psi, \mathbf{i}(\Phi')]_{E'} - [\psi', \mathbf{i}(\Phi)]_{E'} \} - 2 [\Phi,\Phi']_{\g} ),
\end{equation}
where $\mathbf{i} \in \Hom(\g_{P},E')$ is the inclusion and $\mathbf{p} \in \Hom(E', \g_{P})$ is the projection. 

\subsection{Generalized metric and connections relevant for reduction} \label{sec_genmetconrel}
Having a reduction procedure from the Courant algebroid on $E = \gTP$ to $E' = TM \oplus \g_{P} \oplus T^{\ast}M$, we can try to reduce the other ingredients as well. Any generalized metric $\gm$ on $E$ defines an involution $\tau \in \Aut(E)$. Let $\Re: \g \rightarrow \Gamma(E)$ be the map making $E$ into an equivariant Courant algebroid. We assume that $\Re$ and $H$ take the form (\ref{eq_ReHrelevant}). We impose the conditions
\begin{equation}
\tau( [\Re(x), (Y,\eta)]_{E}) = [ \Re(x), \tau(Y,\eta) ]_{E}, \; \; \tau(K^{\perp}) \subseteq K^{\perp}. 
\end{equation}
First condition ensures that $\tau( \Gamma_{G}(E)) \subseteq \Gamma_{G}(E)$. In terms of $(g,B)$, it is equivalent to $g$ and $B$ being the $G$-invariant tensor fields on $P$. To examine the second condition, note that we can restrict $\tau$ to $\Gamma_{G}(E)$ and use the decomposition $\Gamma_{G}(E) \cong \Gamma(E') \oplus \Gamma(\g_{P})$. With respect to this, we have
\begin{equation} \label{eq_taugEblocks}
\tau = \bm{\tau'}{\tau_{1}}{0}{\tau_{\g}}, \; \; g_{E} = \bm{g_{E'}}{0}{0}{-c}. 
\end{equation}
The orthogonality of $\tau$ immediately implies $\tau_{1} = 0$. Moreover, the positive definite fiber-wise metric $\h = -c \tau_{\g}$ on $\g_{P}$ must satisfy $\h c^{-1} \h = c$. For a compact Lie group $G$, there is only one such $\h$, namely $\h = -c$, and consequently $\tau_{\g} = 1$. This observation proves that $\tau( \Re(x)) = \Re(x)$, and consequently $\Re(x) \in \Gamma(V_{+})$. This implies that $g$ and $B$ written as formal two by two block matrices with respect to the decomposition $\mathfrak{X}_{G}(P) \cong TM \oplus \g_{P}$ have the form
\begin{equation}
g = \bm{1}{\vartheta^{T}}{0}{1} \bm{g_{0}}{0}{0}{-\frac{1}{2}c} \bm{1}{0}{\vartheta}{1}, \; \; B =  \bm{B_{0}}{ \frac{1}{2} \vartheta^{T}c}{- \frac{1}{2} c \vartheta}{0},
\end{equation}
for a Riemannian metric $g_{0} > 0$, $B_{0} \in \df{2}$ and $\vartheta \in \Omega^{1}(M,\g_{P})$. But this is precisely the form (\ref{eq_gBrelevant}) required for Kaluza-Klein reduction in Theorem \ref{thm_main}. Moreover, $\tau' \in \Aut(E')$ defines a generalized metric $\gm'$ on $E'$. Plugging in the above form of $g$ and $B$, one finds, not very surprisingly, exactly the generalized metric (\ref{eq_G'metric}). In particular, every generalized metric $\gm'$ can be obtained by the reduction of the unique generalized metric $\gm$. Note that one can write $\gm$ in the block form with respect to the decomposition $\Gamma_{G}(E) \cong \Gamma(E') \oplus \Gamma(\g_{P})$ as
\begin{equation} \label{eq_Gblocks}
\gm = \bm{\gm'}{0}{0}{-c}. 
\end{equation}
Let us now consider a Courant algebroid connection $\cD$ on $E$. First, we impose the condition:
\begin{equation}
[\Re(x), \cD_{\Psi}\Psi']_{E} = \cD_{[\Re(x),\Psi]_{E}} \Psi' + \cD_{\Psi} [\Re(x),\Psi']_{E}, 
\end{equation}
for all $\Psi,\Psi' \in \Gamma(E)$, which ensures that $\cD$ preserves invariant sections. Next, we want $\cD_{\Psi} \Psi' \in \Gamma(K^{\perp})$ whenever $\Psi,\Psi' \in \Gamma(K^{\perp})$. Furthermore, we want $\cD$ to be a Levi-Civita connection on $E$ with respect to $\gm$. It turns out that the most general such $\cD$ has the form
\begin{equation}
\cD_{(\psi,\Phi)} = \bm{\cD'_{\psi}}{V_{\Phi}}{c^{-1} V_{\Phi}^{T} g_{E'}}{D_{\rho'(\psi)} + c^{-1} U(\psi,\star,\cdot) - \frac{2}{3}[\Phi,\star]_{\g}  + c^{-1} \mathfrak{a}(\Phi,\star,\cdot)} ,
\end{equation}
where $\cD' \in LC(E',\gm')$ is a Levi-Civita connection on $E'$ with respect to $\gm'$, $V \in \Omega^{1}(\g_{P}) \otimes \End(\g_{P},V'_{+})$ and $\mathfrak{a} \in \Omega^{1}(\g_{P}) \otimes \Omega^{2}(\g_{P})$ satisfies $\mathfrak{a}(\Phi,\Phi',\Phi'') + cyclic(\Phi,\Phi',\Phi'') = 0$. Finally, $U \in \Omega^{1}(E') \otimes \Omega^{2}(\g_{P})$ is determined by $V$ as 
\begin{equation}
U(\psi,\Phi,\Phi') = \< \psi, V_{\Phi}(\Phi') - V_{\Phi'}(\Phi)\>_{E'} - \< [\Phi,\Phi']_{\g}, \mathbf{p}(\psi) \>_{\g}.
\end{equation}
For the purposes of this paper, we will assume the simplest scenario, where $V = \mathfrak{a} = 0$. The resulting connection reads
\begin{equation} \label{eq_conrelevant}
\cD_{(\psi,\Phi)} = \bm{ \cD'_{\psi}}{0}{0}{ D_{\rho'(\psi)} - [\mathbf{p}(\psi), \star]_{\g} - \frac{2}{3}[\Phi,\star]_{\g}}
\end{equation}
This connection does not contain any additional data except for the Levi-Civita connection $\cD'$ on $E'$. We say that (\ref{eq_conrelevant}) is a Levi-Civita connection on $E$ which is \textbf{relevant for reduction}. 
\subsection{Comparison of Ricci tensors and of scalar curvatures} \label{subsec_comparison}
Let $\cD \in LC(E,\gm)$ be a Levi-Civita connection on $E$ relevant for reduction. The main goal of this subsection to compare its Ricci tensor $\Ric$ and consequently also scalar curvatures to the Ricci tensor $\Ric'$ of the reduced connection $\cD'$. As all quantities are tensorial, we may work with invariant sections and use the decomposition $\Gamma_{G}(E) \cong \Gamma(E') \oplus \Gamma(\g_{P})$. Using (\ref{eq_conrelevant}), it is then a straightforward calculation to prove the following relations: 
\begin{align}
\label{eq_RicasRic'1} \Ric((\psi,0),(\psi',0)) & = \Ric'(\psi,\psi'), \\
\label{eq_RicasRic'2} \Ric((0,\Phi),(0,\Phi')) &= - \frac{1}{6} \<\Phi,\Phi'\>_{\g}, \\
\label{eq_RicasRic'3} \Ric((\psi,0),(0,\Phi')) &= 0.
\end{align}
We slightly abuse the notation, as $\Ric$ is a tensor on $\gTP$. However, evaluated on invariant sections, it defines a function on $M$ and the above relations make sense. Note that one has to use the fact that $\<\cdot,\cdot\>_{\g}$ comes from the Killing form of $\g$, that is 
\begin{equation}
\< \Phi,\Phi'\>_{\g} = \< \Phi^{\alpha}, [\Phi,[\Phi',\Phi_{\alpha}]_{\g}]_{\g} \>, 
\end{equation}
where $\{ \Phi_{\alpha} \}_{\alpha=1}^{\dim{\g}}$ is any local basis for $\Gamma(\g_{P})$.  From (\ref{eq_taugEblocks}, \ref{eq_Gblocks}) it is now easy to find the corresponding relation of the two scalar curvatures. One gets 
\begin{equation} \label{eq_RSrel}
\RS_{\gm} = \RS'_{\gm'} \circ \pi + \frac{1}{6} \dim{\g}, \; \; \RS_{E} = \RS'_{E'} \circ \pi + \frac{1}{6} \dim{\g}. 
\end{equation}
Finally, one can see from (\ref{eq_taugEblocks}) that under the decomposition $\Gamma_{G}(E) \cong \Gamma(E') \oplus \Gamma(\g_{P})$, one has 
\begin{equation}
\Gamma_{G}(V_{+}) = \Gamma(V'_{+}) \oplus \Gamma(\g_{P}), \; \; \Gamma_{G}(V_{-}) = \Gamma(V'_{-}) \oplus \{ 0 \}.
\end{equation}
Together with (\ref{eq_RicasRic'1} - \ref{eq_RicasRic'3}), this proves the following proposition right away: 
\begin{tvrz} \label{tvrz_Riccomp}
Let $\cD \in \LC(E,\gm)$ be a Levi-Civita connection on $E$ relevant for reduction, inducing a Levi-Civita connection $\cD' \in LC(E',\gm')$ on $E'$.

Then $\cD$ is Ricci compatible with $\gm$, if and only if $\cD'$ is Ricci compatible with $\gm'$. 
\end{tvrz}
\section{Proof of the main theorem} \label{sec_proof}
We have now finished the necessary preparations to prove the main theorem of this paper. It is a neat combination of Sections \ref{sec_EOMgeom} and \ref{sec_reduction}. 

\begin{proof}[Proof of Theorem \ref{thm_main}]

Let $\gm$ be a generalized metric (\ref{eq_exgenmetric}) on $E = \gTP$ corresponding to a pair $(g,B)$, and let $\gm'$ be a generalized metric (\ref{eq_G'metric}) on $E' = TM \oplus \g_{P} \oplus T^{\ast}M$ corresponding to a pair $(g_{0},B_{0},\vartheta)$. The assumption on $3$-form $H$ allows one to reduce the Courant algebroid $E$ onto $E'$ as described in the previous section, see the relations (\ref{eq_ReHrelevant}). 

Under the assumptions (\ref{eq_gBrelevant}) of the theorem, $\gm$ reduces to $\gm'$ in the sense of Subsection \ref{sec_genmetconrel}. Let $\cD' \in LC(E',\gm')$ be the connection defined in Theorem \ref{thm_eomheterotic}. 

Now, \emph{define} the connection $\cD \in LC(E,\gm)$ using the block form (\ref{eq_conrelevant}). This determines it uniquely on invariant sections of $E$. As such sections generate $\Gamma(E)$, we extend $\cD$ to all sections via the rules (\ref{eq_conleibniz}). We claim that $\cD$ satisfies the assumptions of Theorem \ref{thm_eomSG} for $\phi = \phi_{0} \circ \pi$. As this involves a slightly non-trivial discussion using the twists of the both Courant brackets, we refer to Section $5$ of our previous paper \cite{Jurco:2015bfs} where we proved this in detail.

Theorem \ref{thm_eomSG} says that $\beta_{g} = \beta_{B} = 0$, if and only if $\cD$ is Ricci compatible with $\gm$. By Proposition \ref{tvrz_Riccomp}, this is equivalent to the Ricci compatibility of $\cD'$ with $\gm'$. Finally, by Theorem \ref{thm_eomheterotic}, this is equivalent to $\beta'_{g_{0}} = \beta'_{B_{0}} = \beta_{\vartheta} = 0$. Next, the combination of Theorem \ref{thm_eomSG}, Theorem \ref{thm_eomheterotic} and the relation (\ref{eq_RSrel}) gives 
\begin{equation}
\begin{split}
\beta_{\phi} = & \ \RS_{\gm} - 2 \Lambda = \RS'_{\gm'} \circ \pi + \frac{1}{6} \dim{\g} - 2 \Lambda \\
= & \ \beta'_{\phi_{0}} \circ \pi  + (2 \Lambda_{0} + \frac{1}{6} \dim{\g}) + \frac{1}{6} \dim{\g} - 2 \Lambda \\
= & \ \beta'_{\phi_{0}} \circ \pi + (\frac{1}{3} \dim{\g} - 2 (\Lambda - \Lambda_{0}) ).
\end{split}
\end{equation}
This is where the last assumption $\Lambda = \Lambda_{0} + \frac{1}{6} \dim{\g}$ becomes important, and we get
\begin{equation} \beta_{\phi} = \beta'_{\phi_{0}} \circ \pi. \end{equation}
We have thus argued that the system (\ref{eq_thmEOM1}) is equivalent to (\ref{eq_thmEOM2}) which by Theorem \ref{thm_EOM} implies the main statement of the theorem. 
\end{proof}
\begin{rem}
We have swept several technical details under the carpet. 

First, Theorem \ref{thm_main} does not assume anything about the signature of the metric $g_{0}$, whereas in the proof, we have worked with the Riemannian metric $g_{0} > 0$. However, the fiber-wise metrics $\gm$ and $\gm'$ remain well-defined and non-degenerate (they are not positive definite), and they still induce the decompositions $E = V_{+} \oplus V_{-}$ and $E' = V'_{+} \oplus V'_{-}$ required to define the Ricci compatibility. All conclusions drawn from the calculations involving Levi-Civita connections thus remain valid. 

Next, if we want to talk about Courant algebroids, the $3$-form $H$ twisting the Dorfman bracket must be closed. However, the form $H = \pi^{\ast}(H_{0}) + \frac{1}{2} CS_{3}(A)$ can be closed only if the first Pontriyagin class $[\< F \^ F\>_{\g}]_{dR}$ vanishes. See \cite{2005math......9563B} for elaborate discussion in the context of Courant algebroids. However, $H$ must be closed only for $E$ and $E'$ to satisfy the full Leibniz identities (\ref{eq_leibnizidentity}). But the only relevant property for all calculations done in this paper is (\ref{eq_rhohom}), which remains valid also for $dH \neq 0$. As we have already noted in Example \ref{ex_dorfman}, we can without issues work with pre-Courant algebroids instead of Courant algebroids. 

Finally, we have assumed compactness of Lie group $G$. This had three implications. First, the Killing form $\<\cdot,\cdot\>_{\g}$ is negative definite. This is necessary when working with Riemannian $g$ due to (\ref{eq_gBrelevant}). If we allow any signature of $g$, this is not needed. Next, the compactness was used in Lemma \ref{lem_equispecial} to find the special form of the map $\Re$. However, we can simply define $\Re$ by (\ref{lem_equispecial}). Finally, for non-compact $G$, the generalized metric $\gm'$ given by (\ref{eq_G'metric}) can be more general, and similarly with the generalized metric $\gm$ relevant for reduction. This poses no problems, as we can always assume $\gm'$ and $\gm$ to be of the required form. Altogether, the compactness assumption on $G$ can be dropped. 
\end{rem}
\section{Analysis of the action $S_{0}$} \label{sec_analysis}
\subsection{Rewriting in terms of local fields}
As a last bit of this paper, let us examine the action (\ref{eq_S0}). We have already argued that $F'$ is actually a curvature $2$-form of a principal bundle connection $A' \in \Omega^{1}(P,\g)$ which is related to the assumed fixed connection $A$ by (\ref{eq_horliftchange}). In the usual treatment of Yang-Mills theory, one uses local sections (equivalent to the choice of gauge) to obtain locally defined objects on $M$, ensuring that the resulting action functional does not depend on the choice of the gauge. 

Let $\sigma: U \rightarrow P$ be a smooth local section of $\pi: P \rightarrow M$, that is in particular $\pi \circ \sigma = 1_{U}$. Let $\Omega \in \Omega^{2}(P,\g)$ and $\Omega' \in \Omega^{2}(P,\g)$ be the curvature $2$-forms of the connections $A$ and $A'$, respectively. Define their local counterparts on $U$ as
\begin{equation}
\A = \sigma^{\ast}(A), \; \; \F = \sigma^{\ast}(\Omega), \; \; \A' = \sigma^{\ast}(A'), \; \; \F' = \sigma^{\ast}(\Omega'). 
\end{equation}
The first term in $S_{0}$ can be on $U$ rewritten using the local form $\F' \in \Omega^{2}(U,\g)$:
\begin{equation}
\dal F' , F' \dar = \< \F' \^ \ast_{g_{0}} \F' \>_{\g} \equiv \dal \F', \F' \dar. 
\end{equation}
This is a standard kinetic term for the non-Abelian gauge field $\A'$ in the Yang-Mills action functional. It is more interesting to examine the $3$-form $H'_{0}$. Note that its restriction onto $U$ is a priori gauge invariant. To find the local expression in terms of the local connection forms, note that for all $X \in \mathfrak{X}(U)$, one has
\begin{equation}
(\A' - \A)(X) = \vartheta(X) \circ \sigma.
\end{equation}
Using this expression, one can replace all occurrences of $\vartheta$ in $H'_{0}$ with the difference $\A' - \A$. The resulting restriction of $H'_{0}$ onto $U$ gives 
\begin{equation} \label{eq_H'0jinak}
H'_{0} = d(B_{0} + \frac{1}{2} \< \A' \^ \A \>_{\g}) + \{ H_{0} + \frac{1}{2} \sC_{3}(\A) \} - \frac{1}{2} \sC_{3}(\A'),
\end{equation}
where $\sC_{3}(\A) \in \Omega^{3}(U)$ is a pullback of the Chern-Simons $3$-form for $A$, that is 
\begin{equation}
\sC_{3}(\A) \equiv \sigma^{\ast}(CS_{3}(A)) = \< \F \^ \A \>_{\g} - \frac{1}{3!} \<[\A \^ \A]_{\g} \^ \A\>_{\g}. 
\end{equation}
The definition of $\sC_{3}(\A')$ is analogous. One can now check explicitly that the expression for $H'_{0}$ is invariant with respect to the usual form of local gauge transformations, that is if one considers instead the section $\~\sigma(m) = \sigma(m) \cdot g(m)$ for all $m \in U$, where $g: U \rightarrow G$ an arbitrary smooth function. One gets
\begin{equation} \label{eq_Agauge}
\~\A = Ad_{g^{-1}}(\A) + g^{\ast}(\theta^{L}),
\end{equation}
where $\theta^{L} \in \Omega^{1}(G,\g)$ is a left Maurer-Cartan form on $G$. A similar rule holds for $\A'$. With some effort, one can show that the local Chern-Simons form changes under this transformation as 
\begin{equation} \label{eq_CSgauge}
\sC_{3}(\~\A) = \sC_{3}(\A) + d ( \< \A \^ g^{\ast}(\theta^{R}) \>_{\g}) - 2 g^{\ast}(\eta),
\end{equation}
where $\theta^{R}$ is the right Maurer-Cartan form on $G$ and $\eta \in \Omega^{3}(G)$ is the canonical biinvariant Cartan $3$-form $\eta = \frac{1}{12} \< [\theta^{L} \^ \theta^{L}]_{\g} \^ \theta^{L} \>_{\g}$. The similar rule applies for $\sC_{3}(\A')$. It is now straightforward to see that the contributions from the gauge transformation cancel and $H'_{0}$ remains unchanged. Note that one often considers only the infinitesimal gauge transformations, where for $\lambda: U \rightarrow \g$, one obtains the following rules,
\begin{equation}
\~A = \A + [\A,\lambda]_{\g} + d\lambda, \; \; \~\F = \F + [\F,\lambda]_{\g}, \; \; \sC_{3}(\~\A) = \sC_{3}(\A) + d( \< d\A, \lambda\>_{\g}),
\end{equation}
obtained from (\ref{eq_Agauge}) and (\ref{eq_CSgauge}) by taking $g = \exp(t\lambda)$ and collecting at most linear terms in the expansion with respect to the infinitesimal parameter $t$. 

To relate this to the usual notation in physics literature, recall that $H_{0} + \frac{1}{2} \sC_{3}(\A)$ is a closed $3$-form on $U$. Assuming that $U$ is contractible, there exists $\B_{0} \in \Omega^{2}(U)$ with $H_{0} + \frac{1}{2} \sC_{3}(\A) = d\B_{0}$. Define a $2$-form $\B \in \Omega^{2}(U)$ as $\B = B_{0} + \frac{1}{2} \<\A' \^ \A\>_{\g} + \B_{0}$ to write $H'_{0}$ (\ref{eq_H'0jinak}) in the form
\begin{equation} \label{eq_H'0usingB}
H'_{0} = d\B - \frac{1}{2} \sC_{3}(\A'). 
\end{equation}
We see that the new field $\B$ must transform non-trivially under the gauge transformations:
\begin{equation}
d\~\B = d\B + \frac{1}{2} d \< \A' \^ g^{\ast}(\theta^{R}) \>_{\g} + g^{\ast}(\eta).
\end{equation}
As $U$ is assumed contractible, for each $g: U \rightarrow G$, we can solve the equation $g^{\ast}(\eta) = d(\xi_{g})$ for some $\xi_{g} \in \Omega^{2}(U)$, which allows us to write the transformation rule for the field $\B$ itself:
\begin{equation} \label{eq_Bgauge}
\~\B = \B + \frac{1}{2} \< \A' \^ g^{\ast}(\theta^{R}) \>_{\g} + \xi_{g}. 
\end{equation}
For the infinitesimal gauge transformation, we obtain a little bit simpler rule $\~\B = \B + \frac{1}{2} \< \A', \lambda\>_{\g}$.

We can now consider the following action functional:
\begin{equation}
S'_{0}[g_{0},\B,\phi_{0},\A'] = \int_{M} e^{-2\phi_{0}} \{ \RS(g_{0}) + \frac{1}{2} \dal \F', \F' \dar - \frac{1}{2} \<\H ,\H\>_{g_{0}} + 4 \< d\phi_{0},d\phi_{0}\>_{g_{0}} - 2 \Lambda_{0} \} \cdot d \vol_{g_{0}},
\end{equation}
where $\H \in \Omega^{3}(M)$ is \textit{defined by} (\ref{eq_H'0usingB}): $\H = d\B - \frac{1}{2} \sC_{3}(\A')$, and $\B$ is any local $2$-form transforming as (\ref{eq_Bgauge}) under the gauge transformations. Note that the only difference between $S_{0}$ and $S'_{0}$ is the fact that our original $\B = B_{0} + \frac{1}{2}\< \A' \^ \A\>_{\g} + \B_{0}$ explicitly depends on $\A'$, and $\B$ in $S'_{0}$ is assumed to be an independent dynamical variable. 

However, this is not an issue, as we will now demonstrate. Consider a variation $\A' \mapsto \A' + \epsilon \cdot \C'$, where $\C'$ vanishes on $\partial M$. Then $H'_{0}$ given by (\ref{eq_H'0jinak}) changes as $H'_{0} \mapsto H'_{0} + \epsilon \cdot \{ \delta \H + \frac{1}{2}  d \< \C' \^ \A\>_{\g} \}$, where $\delta \H$ denotes the first order variation of $-\frac{1}{2} \sC_{3}(\A')$. But the corresponding term in action then changes as
\begin{equation}
e^{-2\phi_{0}} \< H'_{0},H'_{0}\>_{g_{0}} \mapsto e^{-2\phi_{0}} \big( \<H'_{0},H'_{0}\>_{g_{0}} + 2 \epsilon \cdot \<H'_{0}, \delta \H \>_{g_{0}} + \epsilon \cdot \< H'_{0}, d \< \C' \^ \A\>_{\g} \>_{g_{0}} \big).
\end{equation}
The only term arising from the additional dependence of $\B$ on $\A'$ is the last one. But
\begin{equation}
e^{-2 \phi_{0}} \< H'_{0}, d \< \C' \^ \A\>_{\g} \>_{g_{0}} = \< \delta ( e^{-2 \phi_{0}} H'_{0}), \< \C' \^ \A\>_{\g} \>_{g_{0}} = 0,
\end{equation}
where we have used the equation of motion for $B_{0}$ (or $\B$) in the last step. We conclude that it is safe to forget about the variation of $\B$ with $\A'$ and declare its independence. This shows that the actions $S'_{0}$ and $S_{0}$ are classically equivalent. 
\subsection{Relation to heterotic supergravity}
Now, for a particular choice of the principal bundle $P$, one can relate $S'_{0}$ to so called \textbf{heterotic supergravity}. This is a theory obtained as a low-energy limit of the heterotic string. See e.g. \cite{Bergshoeff:1989de} and \cite{2010arXiv1010.4031M}. Without the spin part of the principal bundle, the corresponding effective action is sometimes called Einstein-Yang-Mills gravity, see \cite{Bergshoeff:1981um,Chapline:1982ww}.

Assume that $M$ is a ten-dimensional spin manifold with the Lorentzian metric $g_{0}$ of signature $(9,1)$. Let $\pi_{1}: P_{\textbf{YM}} \rightarrow M$ be any principal bundle with a compact structure group $K$, equal to either $\gSO(32)$ or $\gE(8) \times \gE(8)$. Next, let $\pi_{2}: P_{\textbf{Spin}} \rightarrow M$ be the spinor principal bundle corresponding to the spin structure on $M$, with the structure group $\text{Spin}(9,1)$. We can now consider the fibered product principal bundle $\pi: P \rightarrow M$, where $P = P_{\textbf{YM}} \times_{M} P_{\textbf{Spin}}$ is a principal $(K \times \text{Spin}(9,1))$-bundle. That is $G = K \times \text{Spin}(9,1)$. 

The corresponding Lie algebra $\g$ can be written as a direct product $\g = \ak \oplus \aso(9,1)$. It follows that in this case every connection $A \in \Omega^{1}(P,\g)$ decomposes uniquely as 
\begin{equation}
A = ( p_{1}^{\ast}(A_{\textbf{YM}}), p_{2}^{\ast}(A_{\textbf{Spin}})),
\end{equation}
where $p_{1}: P \rightarrow P_{\mathbf{YM}}$, $p_{2}: P \rightarrow P_{\mathbf{Spin}}$ are the canonical projections, and $A_{\mathbf{YM}} \in \Omega^{1}(P_{\mathbf{YM}}, \ak)$ and $A_{\mathbf{Spin}} \in \Omega^{1}(P_{\mathbf{Spin}}, \aso(9,1))$ are the connections on the respective principal bundles. 

Now, consider the question of rescaling of the invariant form $c = \<\cdot,\cdot\>_{\g}$. One has to carefully keep track of changes in the proof of Theorem \ref{thm_main}. Let $c' = \lambda c$ for non-zero real scalar $\lambda \in \R$. The expressions (\ref{eq_RS'1}, \ref{eq_RS'2}) change almost as expected, that is $c$ is everywhere just replaced by $c'$. Except for the scalars $\frac{1}{6} \dim{\g}$ which change to $\frac{1}{6 \lambda} \dim{\g}$. The same thing happens in the relations (\ref{eq_RSrel}). For $\g = \ak \oplus \aso(9,1)$, we have $\<\cdot,\cdot\>_{\g} = \<\cdot,\cdot\>_{\ak} + \<\cdot,\cdot\>_{\aso}$, where $\<\cdot,\cdot\>_{\ak}$ and $\<x,y\>_{\aso} = 8 \Tr(xy)$ are the Killing forms on $\ak$ and $\aso(9,1)$, respectively. We can then take a more adventurous linear combination $\<\cdot,\cdot\>_{\g} = \lambda_{1} \<\cdot,\cdot\>_{\ak} + \lambda_{2} \<\cdot,\cdot\>_{\aso}$. The resulting scalar is then $\frac{1}{6}( \frac{1}{\lambda_{1}} \dim{\ak} + \frac{1}{\lambda_{2}} \dim{\aso(9,1)})$. To obtain the heterotic supergravity, one chooses $(\lambda_{1},\lambda_{2}) = (-\alpha',\alpha')$, where $\alpha'$ is the usual string parameter corresponding to its tension.  The relation of cosmological constants changes to 
\begin{equation}
\Lambda = \Lambda_{0} + \frac{1}{6 \alpha'}( \dim{\aso(9,1)} - \dim{\ak}) = \Lambda_{0} + \frac{1}{6 \alpha'}( 45 - 496 ).
\end{equation}
The adjoint bundle $\g_{P}$ of $P$ splits as $\g_{P} = \ak_{P_{\mathbf{YM}}} \oplus \aso(9,1)_{P_{\mathbf{Spin}}}$, and the curvature $2$-form $F' \in \Omega^{2}(P,\g_{P})$ thus decomposes accordingly as $F' = (F'_{\mathbf{YM}}, F'_{\mathbf{Spin}})$, and $\F' = \F'_{\mathbf{YM}} + \F'_{\mathbf{Spin}}$ for the local connection $2$-forms. Clearly, the connection $A_{\mathbf{Spin}}$ corresponds to some \textit{spin connection} $\cD^{T}$ on $M$. The kinetic term then decomposes as 
\begin{equation}
\frac{1}{2} \dal \F', \F' \dar = \frac{\alpha'}{2} \{ \dal \F'_{\mathbf{Spin}}, \F'_{\mathbf{Spin}} \dar_{\aso} - \dal \F'_{\mathbf{YM}}, \F'_{\mathbf{YM}} \dar_{\ak} \}.
\end{equation}
Note that in terms of $\cD^{T}$, one finds $\dal \F'_{\mathbf{Spin}}, \F'_{\mathbf{Spin}} \dar_{\aso} = 4 K(\cD^{T})$, where $K(\cD^{T})$ denotes the Kretschmann scalar of the metric connection $\cD^{T}$. The $3$-form $\H$ in $S'_{0}$ can be now written as 
\begin{equation}
\H = d\B + \frac{\alpha'}{2}\{ \sC_{3}(\A'_{\mathbf{YM}}) - \sC_{3}(\A'_{\mathbf{Spin}}) \},
\end{equation}
and thanks to (\ref{eq_pontryiaginvanishes}) it is subject to the so called \textbf{anomaly cancellation condition}:
\begin{equation}
d\H = \frac{\alpha'}{2} \{ \< \F'_{\mathbf{YM}} \^ \F'_{\mathbf{YM}} \>_{\ak} - \< \F'_{\mathbf{Spin}} \^ \F'_{\mathbf{Spin}} \>_{\aso} \}.
\end{equation}
In other words, we find the relation $p_{1}( P_{\mathbf{Spin}}) = p_{1}( P_{\mathbf{YM}})$ of the Pontriyagin classes.
\section*{Acknowledgment}
First and foremost, I would like to thank Braňo Jurčo for his long-time support and ideas allowing me to write this paper. It is a pleasure to thank also Vít Tuček, Urs Schreiber and Pavol Ševera for helpful discussions.
This research was supported by RVO: 67985840, and I would like to thank the Max Planck Institute for Mathematics in Bonn for hospitality. 
\newpage
\bibliography{bib} 
\end{document}